\documentclass[11pt]{article}
\usepackage{enumerate}
\usepackage{amsfonts}
\usepackage{amssymb}
\usepackage{amsmath}
\usepackage{xspace}
\usepackage{algpseudocode}

\usepackage{graphicx}
\usepackage{multicol}
\usepackage[margin=1in]{geometry}
\usepackage{hyperref}
\usepackage{times}

\usepackage{xspace}
\usepackage{amssymb}
\usepackage{amsmath}
\usepackage{amsthm}
\usepackage{latexsym}
\usepackage{graphicx}
\usepackage{color}
\usepackage{dsfont}

\usepackage{hyperref}

\usepackage{amsmath}

\usepackage{centernot}

\newcommand{\G}{{\mathcal G}}%\usepackage{enumitem}

\newcommand{\Exp}{\mathbb{E}}
\newcommand{\payoff}{\text{payoff}}
\newcommand{\MP}{\text{MP}}

\newcommand{\Q}{\mathbb{Q}}
\newcommand{\Nat}{\mbox{I$\!$N}}
\newcommand{\N}{\mbox{I$\!$N}}
\newcommand{\R}{\mbox{I$\!$R}}
\newcommand{\Z}{\mathbb{Z}}

\newcommand{\Min}{{Min}\xspace}
\newcommand{\Max}{{Max}\xspace}
\newcommand{\Killer}{{Cons}\xspace}
\newcommand{\Surv}{{Prod}\xspace}
\newcommand{\thresh}{\textsc{Th}}

\newcommand{\PO}{Player~$1$\xspace}
\newcommand{\PT}{Player~$2$\xspace}
\newcommand{\PLi}{Player~$i$\xspace}
\newcommand{\RT}{\text{RT}}
\newcommand{\play}{\text{play}}
\renewcommand{\path}{\text{path}}
\newcommand{\OCG}{\text{OCG}}

\newcommand{\St}{\mbox{St}}
\newcommand{\Po}{\mbox{Po}}

\newcommand{\stam}[1]{}
\newcommand{\short}[1]{}
\newcommand{\zug}[1]{\langle #1  \rangle}
\newcommand{\set}[1]{\{ #1 \}}
\newcommand{\THRESHBUD}{THRESH-BUDGET\xspace}

\usepackage{xspace}

\usepackage{amsmath}

\usepackage{centernot}

\newtheorem{lemma}{Lemma}
\newtheorem{theorem}[lemma]{Theorem}

\newtheorem{proposition}[lemma]{Proposition}

\newtheorem{xmpl}[lemma]{Example}
\newenvironment{example}{\begin{xmpl}\rm}{\end{xmpl}}

\newtheorem{rmark}[lemma]{Remark}

\theoremstyle{definition}
\newtheorem{definition}[lemma]{Definition}
\usepackage[T1]{fontenc}
\usepackage{authblk}
\title{Infinite-Duration Bidding Games\thanks{This paper is based of the conference publication \cite{AHC17}. This research was supported in part by the Austrian Science Fund (FWF) under grants S11402-N23 (RiSE/SHiNE), Z211-N23 (Wittgenstein Award), and M 2369-N33 (Meitner fellowship).}}
\author[1]{Guy Avni\thanks{guy.avni@ist.ac.at}}
\author[1]{Thomas A. Henzinger\thanks{tah@ist.ac.at}} 
\author[2]{Ventsislav Chonev\thanks{vencho@mpi-sws.org}}
\affil[1]{IST Austria}
\affil[2]{Max Planck Institute for Software Systems (MPI-SWS)}
\date{}
\begin{document}
\maketitle

\begin{abstract}
Two-player games on graphs are widely studied in formal methods as they model the interaction between a system and its environment. The game is played by moving a token throughout a graph to produce an infinite path. There are several common modes to determine how the players move the token through the graph; e.g., in turn-based games the players alternate turns in moving the token. We study the {\em bidding} mode of moving the token, which, to the best of our knowledge, has never been studied in infinite-duration games. The following bidding rule was previously defined and called Richman bidding. Both players have separate {\em budgets}, which sum up to $1$. In each turn, a bidding takes place: Both players submit bids simultaneously, where a bid is legal if it does not exceed the available budget, and the higher bidder pays his bid to the other player and moves the token. The central question studied in bidding games is a necessary and sufficient initial budget for winning the game: a {\em threshold} budget in a vertex is a value $t \in [0,1]$ such that if \PO's budget exceeds $t$, he can win the game, and if \PT's budget exceeds $1-t$, he can win the game. Threshold budgets were previously shown to exist in every vertex of a reachability game, which have an interesting connection with {\em random-turn} games -- a sub-class of simple stochastic games in which the player who moves is chosen randomly. We show the existence of threshold budgets for a qualitative class of infinite-duration games, namely parity games, and a quantitative class, namely mean-payoff games. The key component of the proof is a quantitative solution to strongly-connected mean-payoff bidding games in which we extend the connection with random-turn games to these games, and construct explicit optimal strategies for both players.
\end{abstract}

\section{Introduction}
Two-player infinite-duration games on graphs are an important class of games as they model the interaction between a system and its environment. Questions about the automatic synthesis of a reactive system from its specification \cite{PR89} can be reduced to finding a winning strategy for the ``system'' player in a two-player game. The game is played by placing a token on a vertex in the graph and allowing the players to move it through the graph, thus producing an infinite {\em play}. The qualitative winner or quantitative payoff of the game is determined according to the play. There are several common modes  to define how the players move the token, which are used to model different types of systems. The most well-studied mode is {\em turn-based}, where the vertices are partitioned between the players and the player who controls the vertex on which the token is placed, moves it. Other modes include {\em probabilistic} and {\em concurrent} moves (see \cite{AG11}).

We study {\em bidding} games in which the mode of moving is ``bidding''. Intuitively, in each turn, an auction determines which player moves the token. A concrete bidding rule, which was defined and studied for finite-duration games in \cite{LLPSU99,LLPU96}  is called {\em Richman bidding} (named after David Richman). Both players have budgets, and in each turn a bidding takes place: The players simultaneously submit bids, where a bid is legal if it does not exceed the available budget, the higher bidder pays the other player, and moves the token. Ties can occur and one needs to devise a mechanism for resolving them (e.g., giving advantage to \PO), but our results do not depend on a specific mechanism. 

Bidding arises in many settings that are relevant for several communities within Computer Science, and we list several examples below. In Formal Methods, the players in a two-player game often model concurrent processes. Bidding for moving can model an interaction with a scheduler. The process that wins the bidding gets scheduled and proceeds with its computation. Thus, moving has a cost and processes are interested in moving only when it is critical. Bidding for moving can thus be used to obtain a richer notion of {\em fairness}. When and how much to bid can be seen as quantifying the resources that are needed for a system to achieve its objective. Other takes on this problem include reasoning about which input signals need to be read by the system at its different states \cite{CMH08,AKK15} as well as allowing the system to read chunks of input signals before producing an output signal \cite{HL72,HKT12,KZ14}. Also, our bidding game can model {\em scrip systems} that use internal currencies in order to prevent ``free riding'' \cite{KFH12}; namely, agents who use the resources provided by the system without making their own contribution. Such systems are successfully used in various settings such as databases \cite{SA+96}, group decision making \cite{RSK07}, resource allocation, and peer-to-peer networks (see \cite{JSS14} and references therein). In Algorithmic Game Theory \cite{NRTV07}, auction design is a central research topic that is motivated by the abundance of auctions for online advertisements \cite{Mut09}. Repeated bidding is a form of a sequential auction \cite{LST12b}, which is used in many settings including online advertising. Infinite-duration bidding games can model ongoing auctions and can be used to devise bidding strategies for objectives like: ``In the long run, an advertiser's ad should show at least half of the time''. In Artificial Intelligence, bidding games have been used to reason about combinatorial negotiations \cite{MKT18}. 

Recall that ``bidding'' is a mode of moving and can be studied in combination with any objective. Bidding {\em reachability} games were studied in \cite{LLPU96,LLPSU99}: \PO has a target vertex and an infinite play is winning for him iff it visits the target. The central question that is studied regards a necessary and sufficient budget to guarantee winning, called the {\em threshold budget}. Formally, we assume that the budgets add up to $1$. The threshold budget is a function $\thresh: V \rightarrow  [0,1]$ such that if \PO's budget exceeds $\thresh(v)$ at a vertex $v$, then he has a strategy to win the game from $v$. On the other hand, if \PT's budget exceeds $1-\thresh(v)$, he can win the game from $v$. We illustrate the bidding model and threshold budgets in the following example.

\begin{example}
\label{ex:reach}
Consider the reachability bidding game that is depicted in Figure~\ref{fig:reach}. \PO's goal is to reach $t$, and \PT's goal is to prevent this from happening. What is a necessary and sufficient initial budget for \PO to win from $v_0$? We start with a naive solution by showing that \PO can win if his budget exceeds $0.75$. 
Suppose that the budgets are $\zug{0.75 + \epsilon, 0.25-\epsilon}$, for Player~$1$ and~$2$, respectively, for $\epsilon >0$. In the first turn, \PO bids $0.25$ and wins the bidding since \PT cannot bid above $0.25-\epsilon$. He pays his bid to \PT and moves the token to $v_2$. Thus, at the end of the round, the budgets are $\zug{0.5 + \epsilon,0.5 - \epsilon}$ and the token is placed on $v_2$. In the second bidding, \PO bids all his budget, wins the bidding since \PT cannot bid above $0.5-\epsilon$, moves the token to $t$, and wins the game. 

While an initial budget of $0.75$ suffices for winning, it is is not necessary for winning. We continue to show that the necessary and sufficient budget in $v_0$, i.e., the threshold budget, is $2/3$. That is, we show that for every $\epsilon > 0$, \PO can win with a budget of $2/3 + \epsilon$, and if his initial budget is $2/3-\epsilon$, he loses since \PT can force the game to $v_1$. We show a winning strategy for \PO assuming that the initial budgets are $\zug{2/3+\epsilon, 1/3-\epsilon}$. \PO's bid in the first bidding is $1/3$, which he wins since \PT cannot bid beyond $1/3-\epsilon$, and moves the token to $v_2$. The new budgets are $\zug{1/3+ \epsilon, 2/3 - \epsilon}$. Now, \PO bids $1/3 + \epsilon$. If he wins, he proceeds to $t$ and wins the game. Otherwise, \PT wins the bidding, and moves the token back to $v_0$. Since \PT wins the bidding, he must overbid \PO's bid and pay \PO at least $1/3+\epsilon$. In the worst case, the new budgets are $\zug{2/3+2\epsilon, 1/3-2\epsilon}$. In other words, we are back to $v_0$ only that \PO's budget strictly increases. By continuing in a similar manner, \PO forces his budget to increase by a constant. It will eventually exceed $0.75$ from which he can use the naive solution above to win. The same argument shows that \PO wins with a budget of $1/3+\epsilon$ in $v_2$. Showing that \PT wins from $v_0$ with $1/3+\epsilon$ and from $v_2$ with $2/3+\epsilon$, is dual.

To conclude the example, we note that $\thresh(v_1) = 1$, which intuitively means that even with all the budget, \PO cannot win from $v_1$, and $\thresh(t)=0$, which intuitively means that even with no budget, \PO wins from $t$. \hfill\qed
\end{example}

\begin{figure}[ht]
\begin{minipage}[b]{0.4\linewidth}
\centering
\includegraphics[height=1.5cm]{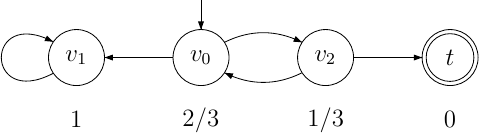}
\caption{A reachability bidding game with the threshold budgets of the vertices.}
\label{fig:reach}
\end{minipage}
\quad
\begin{minipage}[b]{0.4\linewidth}
\centering
\includegraphics[height=1.5cm]{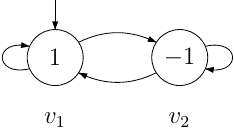}
\caption{A mean-payoff bidding game.}
\label{fig:loops}
\end{minipage}
\end{figure}

It is shown in \cite{LLPU96,LLPSU99} that a threshold budget exists in every vertex of a reachability bidding game. Moreover, it is shown that threshold budgets have the following property: the threshold budget of a vertex $v$ equals $\frac{1}{2}(\thresh(v^+) + \thresh(v^-))$, where $v^+$ and $v^-$ are the successors of $v$ with the maximal and minimal threshold budget, respectively. That is, for every successor $v'$ of $v$, we have $\thresh(v^-) \leq \thresh(v') \leq \thresh(v^+)$. For example, in Example~\ref{ex:reach} we have $\thresh(v_0) = 2/3 = \frac{1}{2}(1 + 1/3) = \frac{1}{2}(\thresh(v_1) + \thresh(v_2))$. 

This property of threshold budgets gives rise to an interesting probabilistic connection. In a {\em random-turn} game, instead of bidding, in each turn, we toss a fair coin. If it turns ``heads'' \PO moves, and if it turns ``tails'', \PT moves. For a reachability bidding game $\G$, we denote by $\RT(\G)$, the random-turn game that is constructed on top of $\G$, which is formally a {\em simple stochastic game}~\cite{Con90} (see Figure~\ref{fig:RT}). It is well-known that every vertex $v$ in $\G$ has a value in $\RT(\G)$, denoted $val(\RT(\G), v)$, which is the probability that \PO wins when both players play optimally. The probabilistic connection for reachability bidding games is the following: for every vertex $v$ in $\G$, $\thresh(v)$ in $\G$ equals $1-val(\RT(\G), v)$. Random-turn based games have been extensively studied in their own right since the seminal paper~\cite{PSSW09}. 

\begin{figure}[ht]
\centering
\includegraphics[height=3cm]{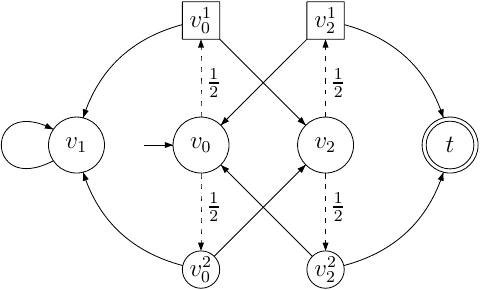}
\caption{The random-turn game that corresponds to the game in Figure~\ref{fig:reach}. The dashed edges are probabilistic and model coin tosses, square vertices are controlled by \PO, and circle vertices by \PT.}
\label{fig:RT}
\end{figure}

We introduce and study infinite-duration bidding games with richer qualitative objectives as well as quantitative objectives. {\em Parity games} are an important class of qualitative games. For example, the problem of reactive synthesis from LTL specifications is reduced to solving a parity game \cite{PR89}. The vertices in a parity game are labeled by an index in $\set{0,\ldots,d}$, for some $d \in \Nat$, and an infinite play is winning for \PO iff the parity of the maximal index that is visited infinitely often is odd. We show that parity bidding games are linearly-reducible to reachability bidding games allowing us to obtain all positive results from these games; threshold budgets exist and the problem of computing them is no harder than for reachability bidding games, which is in turn in NP and coNP due to the probabilistic connection. We find this result somewhat surprising since for most other modes of moving, parity games are considerably harder than reachability games. The key component of the proof considers bottom strongly-connected components (BSCCs, for short) in the game graph, i.e., strongly-connected components with no exiting edges. We show that the BSCCs can be easily classified into those that are ``winning'' for \PO and those that are ``losing'' for him, where in a winning BSCC, \PO wins with any positive initial budget, and in a losing BSCC, \PT wins with any positive initial budget. We can then construct a reachability bidding game by setting the target of \PO to be the winning BSCCs. 
Finally, we ask whether \PO can not only win, but win in a {\em prompt} manner \cite{KPV09}. In {\em B\"uchi} games, which are a special case of parity games, the goal is to visit an accepting vertex infinitely often. We say that \PO wins in a prompt manner if there is a $k \in \Nat$ such that visits to accepting vertices occur within $k$ turns. We show a negative result: under mild assumptions, \PO can never win promptly. That is, with any positive budget, \PT can guarantee arbitrarily long periods with no visits to accepting vertices.

The quantitative games we focus on are {\em mean-payoff} games. The vertices of a mean-payoff game are labeled by weights in $\Z$ and an infinite play has a {\em payoff}, which is the long-run average of the accumulated weights. The payoff is \PO's cost and \PT's reward, thus we refer to the players in a mean-payoff game as {\em Maximizer} (\Max, for short) and {\em Minimizer} (\Min, for short). We adapt threshold budgets to mean-payoff games: we ask what is a necessary and sufficient initial budget  to guarantee a payoff of $0$. We show that threshold budgets exist in mean-payoff bidding games and that finding them is again in NP and coNP.

The key component of the proof, which consists of our most technically challenging result, is a quantitative solution for strongly-connected mean-payoff bidding games by showing an extended probabilistic connection for these games. We show that the optimal payoff \Min can guarantee in a strongly-connected mean-payoff bidding game $\G$ does not depend on his initial budget. More formally, there exists a value $c \in \R$ such that with every positive initial budget, \Min can guarantee a payoff of at most $c$ in $\G$, and he cannot do better: for every $\epsilon >0$ and with any positive budget, \Max can guarantee a payoff that exceeds $c-\epsilon$ in $\G$. Moreover, we show that the optimal payoff $c$ equals the value of the random-turn mean-payoff game $\RT(\G)$. Here, $\RT(\G)$ is a stochastic mean-payoff game and its value is defined as the expected payoff when both players play optimally \cite{MN81}. 

We show a constructive proof for the claim above in which we construct optimal bidding strategies for the two players. Intuitively, the strategies that we construct perform a de-randomization; with a deterministic bidding strategy, the players guarantee that the ratio of the time that is spent in each vertex is the same as in a random behavior. We illustrate our construction in the following example.
%For example, consider the game $\G$ that is depicted in Figure~\ref{fig:loops}. The value of the random-turn game that corresponds to $\G$ is $0$, which intuitively means that in a uniform random process, we expect that the ratio between the number of times that the $+1$ edge is taken is roughly equal to the number of times the $-1$ edge is taken. We construct a deterministic bidding strategy that ensures roughly the same behavior in the bidding game; namely, in every play, the ratio between the two edges is roughly the same.  
Technically, consider an infinite play $\pi$. The {\em energy} of a prefix $\pi^n$ of length $n$ of $\pi$, denoted $E(\pi^n)$, is the sum of the weights that it traverses. The payoff of $\pi$ is $\lim\inf_{n\to \infty} E(\pi^n)/n$.  Note that the definition favors \Min. The strategy we construct for \Min guarantees that an infinite play $\pi$ either has (1) infinitely many prefixes with $E(\pi^n) = 0$, or (2) the energy is eventually bounded, thus there is $N \in \Nat$ such that, after some point $M$, for every $n \in \Nat$ with $n > M$, we have $E(\pi^n) \leq N$. It is not hard to see that this property implies that the payoff of $\pi$ is non-positive. We stress the point that there are two ``currencies'' in the game: the players' budgets are ``monopoly money'' that they do not care about, rather a player's goal is to optimize the payoff, which arises from the weights that are traversed by the play. 

\stam{
We show two proofs for the claim above. The first proof is purely existential. It relies on results on {\em one-counter stochastic games} \cite{BB+10,BBEK11}, which in turn, are equivalent to discrete {\em quasi-birth-death processes} \cite{EWY10} and generalize {\em solvency games} \cite{BKSV08}. The existential proof can be used to determine which player is the winner in a strongly-connected game. It does not, however, provide any insight on how to construct a winning bidding strategy for a player. This is particularly lacking in bidding games, where finding the right bids is the most important and challenging part of constructing a strategy.

We show a different constructive proof for the claim above in which we construct winning strategies for \Min and \Max in strongly-connected mean-payoff bidding game. As we elaborate later in this section, the ideas that were developed in this construction were later used to solve mean-payoff games with other bidding rules, were an existential proof was not possible. We illustrate some of the ideas in the construction of the strategies in the following simple game.
}

\stam{
\begin{figure}[ht]
\begin{minipage}
\begin{center}
\includegraphics[height=1cm]{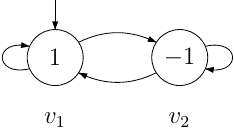}
\caption{A mean-payoff bidding game where the weights are depicted on the edges.}
\label{fig:loops}
\end{center}
\end{figure}
}

\begin{example}
\label{ex:tit-for-tat}
Consider the mean-payoff bidding game $\G$ that is depicted in Figure~\ref{fig:loops}. The value of the random-turn game that corresponds to $\G$ is $0$. Indeed, in $\RT(\G)$, \Min always proceeds to $v_2$ and \Max always proceeds to $v_1$. Since the players are selected to move uniformly at random, the game can be seen as a random walk that takes each edge with probability $0.5$ and stays, in the long run, in $v_1$ and in $v_2$ the same portion of the time. We claim that \Min has a deterministic strategy that guarantees a non-positive payoff. It intuitively guarantees that an infinite play stays in $v_2$ for at least half the time. Without loss of generality, \Max always proceeds to $v_1$ upon winning a bidding. \Min's strategy is a {\em tit-for-tat}-like strategy, and he always proceeds to $v_2$ upon winning a bidding. 

The difficulty is in finding the right bids. \Min maintains a queue. When the queue is empty, \Min bids $0$. If the queue is not empty, \Min bids the smallest element in the queue, and removes it upon winning a bidding. If \Max wins a bidding with $b$, then \Min adds $b$ to the queue. For example, suppose \Max bids $\frac{1}{3}, \frac{1}{2}, \frac{1}{6}$ in the first three biddings. \Min's first bidding is $0$, he loses, and adds $\frac{1}{3}$ to the queue. In the second bidding, \Min bids the minimal element $\frac{1}{3}$ in the queue, loses again, and adds $\frac{1}{2}$ to the queue. In the third bidding, \Min wins with his bid of $\frac{1}{3}$, removes it from the queue, and his bid in the fourth bidding is $\frac{1}{2}$. For simplicity, we assume that \Min wins whenever a tie occurs.

We claim that the tit-for-tat strategy guarantees a non-positive mean-payoff value. Intuitively, elements in \Min's queue can be thought of as \Max winnings that are not ``matched'' by a \Min win. Thus, if the size of the queue is $k$, the energy is at most $k$ (this is an upper bound since \Min could win with $0$ bids). In particular, if the queue is empty, the energy is at most $0$. Suppose the minimal element in the queue is $b$. Then, we claim that the size of the queue, and in turn the accumulated energy, is at most $\lceil 1/{b} \rceil$. Indeed, since each bid in the queue represents an ``unmatched'' \Max bid, if the queue size is greater than $\lceil 1/{b} \rceil$, then the sum of \Max's winning bids is more than $1$, which is impossible since he would need to invest more than the total budget. It follows that \Min's strategy guarantees that in an infinite play either (1) the queue empties infinitely often, thus the energy hits $0$ infinitely often, or (2) if there is a point after which the queue stays non-empty, then its size is bounded, hence the energy is bounded. By the above, this property implies a non-positive payoff.\hfill\qed
\end{example}

As the tit-for-tat strategy above demonstrates, the strategies that we construct carefully match changes in budget with changes in energy. The first step in our construction for general strongly-connected games is to assign an ``importance'' to each vertex in the game; the more important a vertex is, the higher a player bids in it. Our definition of importance uses the concept of {\em potentials} in stochastic games (see \cite{Put05}), which were initially used in the context of the strategy iteration algorithm \cite{How60}. In the second component of the proof, we find a bid by carefully normalizing the importance of a vertex. Normalization is easier in \Min's case because of the asymmetry in the definition of payoff. As demonstrated in the tit-for-tat strategy, \Min keeps the energy bounded from above. \Max strategy guarantees that the energy is bounded from below, which is more technically challenging to achieve.

\paragraph*{Results on other bidding mechanisms} 
Since the first publication of this work, further results were obtained on infinite-duration bidding games with other bidding mechanisms. A second bidding rule that was first defined in \cite{LLPSU99} is called {\em poorman} bidding: the winner of a bidding, rather than paying his bid to the loser, pays the bid to the ``bank'', thus the sum of budgets decreases as the game proceeds. Poorman bidding naturally model settings in which the scheduler accepts payment such as miners in block-chain technology or the auctioneer in ongoing auctions. The mathematical structure of reachability poorman-bidding games is more involved than with Richman bidding. Namely, no probabilistic connection is known and it is unlikely to exist. 

Given the probabilistic connection for reachability Richman-bidding games, the probabilistic connection for mean-payoff Richman-bidding games may not be unexpected. The ideas that were developed in the constructions we show here were later used to show a surprising probabilistic connection for mean-payoff poorman-bidding games in \cite{AHI18}, which is in fact richer than the one we observe here for mean-payoff Richman-bidding games. Then, to better understand the curious differences between the seemingly similar bidding rules, infinite-duration bidding games with {\em taxman} bidding are studied in \cite{AHZ19Arxiv}. Taxman bidding, which was also defined in \cite{LLPSU99} and studied for reachability games, span the spectrum between Richman and poorman bidding. A probabilistic connection was shown for these games as well. We elaborate on these results in Section~\ref{sec:poorman}.

\paragraph*{Further related work on bidding games}
Motivated by recreational games, e.g., bidding chess, {\em discrete bidding games} are studied in \cite{DP10}, where the granularity of the bids is bounded by dividing the money into chips. The Richman calculus for  reachability continuous-bidding games is extended to discrete-bidding in \cite{DP10}. Unlike in continuous-bidding, ties play a crucial role in discrete-bidding. The question of which tie-breaking mechanism gives rise to {\em determinacy}  in infinite-duration discrete-bidding games is investigated in \cite{AAH19Arxiv}.  Non-zero-sum two-player games were studied in \cite{MKT18}. They consider a bidding game on a directed acyclic graph. Moving the token through the graph is done by means of bidding. The game ends once the token reaches a sink, and each sink is labeled with a pair of payoffs for the two players that do not necessarily sum up to $0$. They show existence of {\em subgame perfect equilibrium} for every initial budget and a polynomial algorithm to compute it.

\section{Preliminaries}
%\paragraph*{Bidding games}
A graph game is played on a directed graph $G = \zug{V, E}$, where $V$ is a finite set of vertices and $E \subseteq V \times V$ is a set of edges. The {\em neighbors} of a vertex $v \in V$, denoted $N(v)$, is the set of vertices $\set{u \in V: \zug{v,u} \in E}$. We say that $G$ has out-degree $2$ if for every $v \in V$, we have $|N(v)| = 2$. A {\em path} in $G$ is a finite or infinite sequence of vertices $\eta = v_1,v_2,\ldots$ such that for every $i \geq 1$, we have $\zug{v_i,v_{i+1}} \in E$. A {\em strongly-connected component} of $G$ is a set of vertices $S$ such that for every $u,v \in S$ there is a path from $u$ to $v$ in $G$. A {\em bottom strongly-connected component} (BSCC, for short) is a maximal strongly-connected component $S$ that has no outgoing edges, i.e., there are no edges of the form $\zug{v,u}$, where $v \in S$ and $u \notin S$.

\paragraph{Bidding for moving}
A graph game is a two-player game, which proceeds by placing a token on a vertex in a graph and letting the two players move it to produce an infinite {\em play}. The play gives rise to a path that determines the qualitative winner or quantitative payoff of the game. We refer to the mechanism that determines how the token moves as the {\em mode of moving} of the game. For example, the simplest and most well-studied mode of moving is {\em turn-based}; the vertices are partitioned between the two players and the player who controls the vertex on which the token is placed, moves it.

We study a different mode of moving, which we call {\em bidding}. Both players have budgets, where for convenience, we have $B_1 + B_2 = 1$. In each turn, a bidding takes place to determine which player moves the token: Both players simultaneously submit bids, where a bid is a real number in $[0,B_i]$, for $i \in \set{1,2}$, the player who bids higher pays the other player and moves the token. Note that the sum of budgets always remains $1$. While draws can occur, our results are not affected by the tie-breaking mechanism that is used. To simplify the presentation, we fix the tie-breaking mechanism to always give advantage to \PO.

\paragraph{Strategies and plays}
A {\em strategy} is a recipe for how to play a game. It is a function that, given a finite {\em history} of the game, prescribes to a player which {\em action} to take, where we define these two notions below. For example, in turn-based games, a strategy takes as input, the sequence of vertices that were visited so far, and it outputs the next vertex to move to. In bidding games, histories and strategies are more involved since they maintain the information about the bids and winners of the bids. Formally, a history in a bidding game is $\pi = \zug{v_1, b_1, i_1}, \ldots, \zug{v_k, b_k, i_k}, v_{k+1} \in (V \times \R \times \set{1,2})^*\cdot V$, where for $1 \leq j \leq k+1$, the token is placed on vertex $v_j$ at round $j$, for $1 \leq j \leq k$, the winning bid is $b_j$ and the winner is Player~$i_j$. Consider a finite history $\pi$. For $i \in \set{1,2}$, let $W_i(\pi) \subseteq \set{1,\ldots, k}$ denote the indices in which \PLi is the winner of the bidding in $\pi$. We denote by $B_i(\pi)$ \PLi's budget following $\pi$. Let $B^I_i$ be the initial budget of \PLi. \PO's budget following $\pi$ is $B_1(\pi) = B^I_i - \sum_{j \in W_1(\pi)} b_j + \sum_{j \in W_2(\pi)} b_j$, and \PT's budget is defined dually. Given a history $\pi$ that ends in $v$, a strategy for \PLi prescribes an action $\zug{b, v}$, where $b \leq B_i(\pi)$ is a bid that does not exceed the available budget and $v$ is a vertex to move to upon winning, where we require that $v$ is a neighbor of $v_{k+1}$. 

An initial vertex $v_1$, initial budgets, and two strategies $f_1$ and $f_2$ for the players determine a unique infinite {\em play} for the game, which we denote by $\play(v_1, f_1, f_2)$, and we define its prefixes inductively. Let $\pi^1 = v_1$. Assume that for $n \geq 1$, we have defined the prefix $\pi^n = \zug{v_1, b_1, i_1}, \ldots, \zug{v_k, b_k, i_n}, v_{n+1}$, and we define the prefix $\pi^{n+1}$. Let $\zug{b', v'}= f_1(\pi^n)$ and $\zug{b'', v''} = f_2(\pi^n)$ be the two actions proposed by the two players' strategies. Then, if $b' \geq b''$, \PO wins and we define $\pi^{n+1}=\zug{v_1,b_2,i_1},\ldots,\zug{v_{k+1}, b', 1}, v'$. If $b'' > b'$, \PT wins, and we define $\pi^{n+1} = \zug{v_1,b_2,i_1},\ldots,\zug{v_{k+1}, b'', 2}, v''$. The path that $\play(v_1, f_1, f_2)$ traverses is $\path(\play(v_1, f_1, f_2)) = v_1,v_2,\ldots$.

\stam{
\paragraph*{Strategy complexity}
We seek to find ``simple'' strategies for the players. In stochastic games, a {\em positional} strategy is a strategy that always proceeds the same in a vertex. That is, a positional strategy $f_i$ for \PLi, for $i \in \set{1,2}$, is a function $f_i: V_i \rightarrow V$; whenever the token reaches a vertex $v \in V_i$ that is controlled by \PLi, he always proceeds to $f_i(v)$. 

We extend the definition of positional strategies to bidding games, though the right definition is not immediate. A first attempt is to define a positional strategy as a function from vertex and budget to bid and action. But, since budgets are possibly real numbers, it is not clear how such a strategy is implemented. To overcome this issue, we define a positional strategy in a vertex $v \in V$ as a pair $\zug{u, f_v}$, where $u \in adj(v)$ is the vertex to proceed to upon winning and $f_v: [0,1] \rightarrow [0,1]$ is a function that takes the current budget and returns a bid. We require that $f_v$ is {\em simple}, namely a polynomial or a selection between a constant number of polynomials.  A {\em constant memory} 
}

\paragraph{Objectives} 
An objective $O$ is a set of infinite paths. \PO wins an infinite play $\pi$ iff $\path(\pi) \in O$. We call a strategy $f$ {\em winning} for \PO from a vertex $v$ w.r.t. an objective $O$ if for every strategy $g$ of \PT $\play(v, f, g)$ is winning for \PO. Winning strategies for \PT are defined dually. We consider the following qualitative objectives:
\begin{enumerate}
    \item In \textit{reachability games}, \PO has a target vertex $t$ and an infinite play is winning iff it visits $t$. We sometimes use a set of vertices $T$ as the target of \PO, then \PO wins iff a vertex in $T$ is visited.
    \item In \textit{parity games}, each vertex is labeled with an index in $\{1,\dots,d\}$. An infinite path is winning for \PO iff the parity of maximal index visited infinitely often is odd.
    \item \textit{Mean-payoff games} are played on weighted directed graphs, with weights given by a function $w:V \rightarrow \Q$. Consider an infinite path $\eta =v_1,v_2,\dots \in V^\omega$. For $n \in \Nat$, the prefix of length $n$ of $\eta$ is $\eta^n$, and we define its {\em energy} to be $E(\eta^n)=\sum_{i=1}^{n} w(v_i)$. The {\em payoff} of $\eta$ is $\payoff(\eta) = \liminf_{n\rightarrow \infty}E(\eta^n)/n$. \PO wins $\eta$ iff $\payoff(\eta) \geq 0$.
\end{enumerate}

Mean-payoff games are quantitative games. We think of the payoff as \PO's reward and \PT's cost, thus in mean-payoff games, we refer to \PO as \Max and to \PT as \Min. We elaborate on the quantitative solution to mean-payoff games in Section~\ref{sec:MP}.

\paragraph{Threshold budgets}
The first question that arises in the context of bidding games asks what is the necessary and sufficient initial budget to guarantee an objective. We generalize the definition in~\cite{LLPSU99,LLPU96}:

\begin{definition}\label{def:thresh} ({\bf Threshold budgets})
Consider a bidding game $\G$, a vertex $v$, and an objective $O$ for \PO. The threshold budget in $v$, denoted $\thresh(v)$, is a number in $[0,1]$ such that for an initial budget $B_I \in [0,1]$ for \PO we have
\begin{itemize}
\item if $B_I > \thresh(v)$, then \PO has a winning strategy that guarantees $O$ is satisfied, and 
\item if $B_I < \thresh(v)$, then \PT has a winning strategy that violates $O$.
\end{itemize}
\end{definition}

\paragraph*{Random-turn games}
A {\em stochastic game} is played on an {\em arena} $\zug{V_1, V_2, V_N, \Delta, \Pr}$, where for $i \in \set{1,2}$, $V_i$ is a set of vertices that is controlled by \PLi, $V_N$ is a set of probabilistic vertices that is controlled by ``Nature'', where all three sets are disjoint and we denote $V = V_1 \cup V_2 \cup V_N$, $\Delta \subseteq (V_1 \cup V_2) \times V$ is a set of deterministic edges, and $\Pr: V_N \times V \rightarrow [0,1]$ are probabilistic transitions, i.e., for each $v \in V_N$, we have $\sum_{u \in V} \Pr[v,u] = 1$. As in turn-based games, whenever the game reaches a vertex in $V_i$ that is controlled by \PLi, for $i \in \set{1,2}$, he choses how the game proceeds, and whenever the game reaches a vertex $v \in V_N$, the next vertex is chosen probabilistically according to $\Pr$.

Consider a bidding game $\G$ that is played on a graph $\zug{V, E}$. The {\em random-turn game} that is associated with $\G$ is a stochastic game that intuitively simulates the following process. In each turn we throw a fair coin. If it turns ``heads'', then \PO moves the token, and \PT moves if the coin turns ``tails''. See an example in Figure~\ref{fig:RT}. Formally, we define $\RT(\G) = \zug{V_1, V_2, V, \Delta, \Pr}$, where we make two additional copies of each vertex in $V$; for $i \in \set{1, 2}$, we have $V_i = \set{v_i: v \in V}$. Nature vertices simulate the coin toss: for $v \in V$, we have $\Pr[v, v_1] = \Pr[v, v_2] = 1/2$. Reaching a vertex $v_i \in V_i$, for $i \in \set{1,2}$, means that \PLi won the coin toss and gets to choose a neighbor $u \in N(v)$ to move the token to, thus we have $\Delta = \set{\zug{v_i, u}: \zug{v,u} \in E \text{ and } i \in \set{1, 2}}$. 

The objective of \PO in $\RT(\G)$ is the same as his objective in $\G$. When $\G$ is a reachability game, then $\RT(\G)$ is called a {\em simple stochastic game} \cite{Con92} and the target is the same as in $\G$. When $\G$ is a mean-payoff game, then $\RT(\G)$ is a stochastic mean-payoff game. The weight of $v_1,v_2$, and $v$ all equal the weight of $v$ in $\G$. 

The following definitions are standard, and we refer the reader to \cite{Put05} for more details. Two strategies $f_1$ and $f_2$ for the two players and an initial vertex $v$ give rise to a probability distribution $D(v, f_1, f_2)$ over infinite paths that start in $v$.

\begin{definition}
({\bf Values in stochastic games})
Consider a stochastic game $\G$. When $\G$ is a qualitative game with objective $O$, the {\em value} in a vertex $v$ in $\G$, denoted $val(\G, v)$, is $\sup_{f_1} \inf_{f_2} \Pr_{\eta \sim D(v, f_1, f_2)}[\eta \in O]$. When $\G$ is a mean-payoff game, the value in $v$ is $\sup_{f_1} \inf_{f_2} \Exp_{\eta \sim D(v, f_1, f_2)}[\payoff(\eta)]$. For the objectives we consider, positional optimal strategies exist.
\end{definition}

The existence of positional optimal strategies implies that by letting \PT choose his strategy before \PO, i.e., switching the order in the definitions to $\inf_{f_2} \sup_{f_1}$, we obtain the same value. Moreover, restricting one or both of the players to use only positional strategies does not change the value.

\section{Qualitative Bidding Games}
We start by surveying the results of \cite{LLPU96,LLPSU99} on reachability games before moving to study parity bidding games. The model that is studied in \cite{LLPU96,LLPSU99} uses a slightly different definition of reachability games, which we call {\em double-reachability} games: both players have a target, which we denote by $v_R$ and $v_S$ for ``reach'' and ``safe'', and the game ends once one of the targets is reached. We assume that all vertices apart from $v_R$ and $v_S$ have at least one path to both $v_R$ and $v_S$. We later show that reachability bidding games are equivalent to double-reachability bidding games. 

\begin{theorem}
\label{thm:reach}
\cite{LLPU96,LLPSU99}
Consider a double-reachability bidding game $\G = \zug{V,E, v_R, v_S}$. Then $\thresh(v_R) = 0$ and $\thresh(v_S) = 1$, and for all other vertices $v \in V \setminus \set{v_R, v_S}$, we have $\thresh(v) = \frac{1}{2}\big( \thresh(v^+) + \thresh(v^-)\big)$, where $v^-, v^+ \in N(v)$ are such that for every $v' \in N(v)$, we have $\thresh(v^-) \leq \thresh(v') \leq \thresh(v^+)$. Moreover, for every vertex $v \in V$, we have $\thresh(v) = 1-val(\RT(\G), v)$.
\end{theorem}
\begin{proof}
We describe the key ideas in the proof for completeness. Consider two optimal memoryless strategies $f_1$ and $f_2$ in $\RT(\G)$. For $v \in V \setminus \set{v_S, v_R}$, we define $v^-, v^+ \in N(v)$ according to these strategies: let $v^-= f_1(v_1)$ and $v^+ = f_2(v)$. It is not hard to see that $val(\RT(\G), v_S) = 0$ and $val(\RT(\G), v_R) = 1$, and for every $v \in V \setminus \set{v_S, v_R}$, we have $val(\RT(\G), v) \in (0,1)$ and $val(\RT(\G), v) = \frac{1}{2}\big(val(\RT(\G), v^+) + val(\RT(\G), v^-)\big)$. We claim that if \PO's budget $B_I$ at $v \in V$ exceeds $1-val(\RT(\G), v)$, then he wins the game. Thus, we show that $\thresh(v) \leq 1-val(\RT(\G), v)$. The proof for the other direction is dual.

For $v \in \set{v_S, v_R}$, the claim is trivial. Let $B(v) = 1-val(\RT(\G), v)$ and $\epsilon>0$  be \PO's ``surplus'', namely $B_I = B(v) + \epsilon$. Intuitively, \PO's strategy ensures that he either wins the game or his surplus increases by a constant. For $u \in V \setminus \set{v_S, v_R}$, let $b(u) = \frac{1}{2}\big(val(\RT(\G), u^-) - val(\RT(\G), u^+)\big)$, which by this theorem, is equivalent to $\frac{1}{2}(\thresh(v^+) - \thresh(v^-))$. For example, in the game that is depicted in Figure~\ref{fig:reach}, we have $\thresh(v_1) = 1$, $\thresh(v_2) = 1/3$, and $b(v_0) = (1-1/3)/2 = 1/3$.

Let $n = |V|$ and, for $1 \leq i \leq n$, we define $\epsilon_i = \epsilon\cdot 2^{-i}$. Until he loses a bidding, assuming $0 \leq j \leq n-1$ turns have passed and the token is placed on a vertex $u \in V$, \PO bids $b(u) + \epsilon_{n-j}$ and proceeds to $u^-$ upon winning. Thus, \PO's bid consists of two parts: the major part is $b(u)$ and the minor part is $\epsilon_{n-j}$. We show that no matter the outcome of the bidding, assuming the game continues to $u'$, \PO's budget exceeds $B(u')$. If \PO wins the bidding, he moves the token to $u^-$. His budget exceeds $B(u^-)$ since  $B(u) - b(u) = B(u^-)$ and $\sum_{1 \leq i \leq n} \epsilon_i \leq \epsilon$. On the other hand, if \PT wins, in the worst case, he proceeds to $u^+$ and pays \PO at least $b(u) + \epsilon_{n-j}$. Then, \PO's budget exceeds $B(u^+)$ since $B(u) + b(u) = B(u^+)$ and $\sum_{0 \leq \ell < i} \epsilon_{n-\ell} < \epsilon_{n-j}$. Moreover, the difference in the inequality is at least $\epsilon_n$. Since $B(v_S) = 1$ and the total budget is $1$, the invariant implies that \PT cannot win the game. It is not hard to show that if \PO wins $n$ biddings, he wins the game. Finally, suppose \PO loses the $i$-th bidding, for $1 \leq i \leq n$, then his surplus increases by at least $\epsilon \cdot 2^{-n}$. By repeatedly following this strategy, if \PO does not win the game, his budget will eventually be close enough to $1$, and he can switch to a strategy that wins $n$ biddings in a row and proceeds to $v_R$ on the shortest path.
\end{proof}

\begin{example}
Consider the reachability game $\G$ that is depicted in Figure~\ref{fig:reach} and $\RT(\G)$ that is depicted in Figure~\ref{fig:RT}. The optimal strategy for \PO in $\RT(\G)$ proceeds from $v^1_0$ to $v_2$ and from $v^1_2$ to $t$. The optimal strategy for \PT proceeds from $v^2_2$ to $v_0$ and from $v^2_0$ to $v_1$. We have $1-val(\RT(\G), v_1) = 1 = val(\RT(\G), t)$, $val(\RT(\G), v_0) = \frac{1}{2}(val(\RT(\G), v_2) + val(\RT(\G), v_1)$, and $val(\RT(\G), v_2) = \frac{1}{2}(val(\RT(\G), v_0) + val(\RT(\G),t)$. Note that for every vertex $v$, we have $val(\RT(\G), v) = 1-\thresh(v)$. For example, $val(\RT(\G), v_0) = 1/3 = 1-\thresh(v_0)$.\hfill\qed
\end{example}

The following proposition makes the equivalence between reachability and double-reachability bidding games precise. Also, unlike Theorem~\ref{thm:reach}, it handles  double-reachability games in which only one player has a target, which will be important in our solution to infinite-duration games. Consider a reachability bidding game $\G = \zug{V, E, t}$. Let $S \subseteq V$ be the set of vertices that have no path to $t$, and let $T \subseteq V$ be the set of vertices that have no path to vertices in $S$. Note that every vertex in $T$ has a path to $t$ but there might be cycles that are contained in $T$ that do not cross $t$. Intuitively, we construct a double-reachability game from $\G$ by merging the vertices in $T$ into the target of \PO and merging the vertices in $S$ into the target of \PT. Formally, the double-reachability game that corresponds to $\G$ is $\text{DR}(\G) = \zug{V', E', v_R, v_S}$, where $V' = (V \setminus (S \cup T))\cup \set{v_R, v_S}$, $E' = E \cap (V' \times V') \cup \set{\zug{v,v_R}: \exists u\in T \text{ s.t. } \zug{v,u} \in E} \cup \set{\zug{v,v_S}: \exists u\in S \text{ s.t. } \zug{v,u} \in E}$.

\begin{proposition}
\label{prop:Richman-Reach}
Consider a reachability bidding game $\G = \zug{V, E, t}$. Let $S \subseteq V$ be the set of vertices that have no path to $t$, and let $T \subseteq V$ be the set of vertices that have no path to vertices in $S$. For every $v \in S$, we have $\thresh(v) = 1$, for every $v \in T$, we have $\thresh(v) = 0$, and for every $v \in V \setminus (S \cup T)$, we have that $\thresh(v)$ in $\G$ equals $\thresh(v)$ in the double-reachability game $\text{DR}(\G)$ that corresponds to $\G$.
\end{proposition}
\begin{proof}
The claim that $\thresh(v) = 1$, for every $v \in S$ is trivial. The proof for $v \in T$ is similar to Theorem~\ref{thm:reach}. Suppose \PO's budget at $v$ is $\epsilon > 0$. Assuming $|T| = n$, he chooses $\epsilon_1>\ldots>\epsilon_n$, and, for $0 \leq j \leq n-1$, he bids $\epsilon_{n-j}$ in the $j$-th bidding. Upon winning a bidding, he proceeds to a vertex that is closer to $t$ than the current vertex. As in the proof of Theorem~\ref{thm:reach}, if \PO wins $n$ biddings, he wins the game, and if he loses a bidding, his budget increases by at least $\epsilon_n$. By repeatedly following this strategy, his budget will eventually suffice for winning $n$ biddings in a row. 

The final case is when $v$ is in $V \setminus (S \cup T)$. Suppose \PO starts in $\G$ with a budget of $\thresh(v) + \epsilon$ in $\text{DR}(\G)$. \PO acts as if his budget is $\thresh(v) + \epsilon/2$ and uses the winning strategy in $\text{DR}(\G)$ to force the game to a vertex in $T$. Then, he uses the strategy above with an initial budget of $\epsilon/2$ to force the game to $t$.
\end{proof}

We proceed to study parity bidding games.

\begin{theorem}
\label{thm:parity}
Parity bidding games are linearly reducible to reachability bidding games. Thus threshold budgets exist in parity bidding games.
\end{theorem}
\begin{proof}
Consider a parity bidding game $\G = \zug{V, E, p}$ and let $S$ be a BSCC in $\G$. We claim that there is $\alpha_S \in \set{0,1}$ such that for every $v \in S$, we have $\thresh(v) = \alpha$. In case of $\alpha=0$, we call $S$ ``winning'' for \PO, and when $\alpha = 1$, we call $S$ ``losing'' for \PO. Let $v \in S$ be the vertex with the maximal parity in $S$. We claim that $S$ is winning for \PO iff $p(v)$ is odd. Suppose $p(v)$ is odd, and the other case is dual. We show that \PO can win from $v$ with an initial budget of $\epsilon > 0$. Proposition~\ref{prop:Richman-Reach} implies that \PO can force the game from any vertex in $S$ to $v$ with any positive initial budget. Indeed, we construct a reachability bidding game $\G_S$ by restricting $\G$ to $S$ and setting the target of \PO to be $v$. Since $S$ is a BSCC, there is no vertex from which there is no path to $v$, thus the proposition implies that the threshold budgets are all $0$. Finally, \PO splits $\epsilon$ into infinitely many pieces $\epsilon_1,\epsilon_2,\ldots$, by defining $\epsilon_i = \epsilon\cdot 2^{-i}$, for $i \geq 1$. Initially, he plays as if his budget is $\epsilon_1$ and forces the game to visit $v$ using the strategy in the reachability game. Once $v$ is visited, he repeats the strategy with an initial budget of $\epsilon_2$. He continues similarly forcing infinitely many visits to $v$. Since $v$ is the vertex with maximal parity in $S$ and it is odd, the strategy guarantees that \PO wins.

We now consider vertices in $V$ that are not in a BSCC. Let $W, L \subseteq V$ be the sets of vertices in $V$ that belong to winning and losing BSCCs for \PO, respectively. Let $W'$ and $L'$ be the sets of vertices with no path to vertices in $L$ and $W$, respectively. Note that $W \subseteq W'$ and $L \subseteq L'$ and as in the above, for every $v \in W'$, we have $\thresh(v) = 0$, and for every $v \in L'$, we have $\thresh(v) = 1$. We construct a double-reachability bidding game $\text{DR}(\G)$ by setting the target for \PO to be $W'$ and the target for \PT to be $L'$. A similar argument to the one above shows that for every $v \in V \setminus (W' \cup L')$, $\thresh(v)$ in $\G$ equals $\thresh(v)$ in $\text{DR}(\G)$, and we are done.
\end{proof}

We adress the computational complexity of finding threshold budgets. We first phrase the problem as a decision problem.
\begin{definition}
The input to the \THRESHBUD problem is a parity bidding game $\G$ and a vertex $v$, and the goal is to decide whether $\thresh(v) \geq 1/2$. 
\end{definition}

It is stated in \cite{LLPSU99} that \THRESHBUD is in NP and not known to be in P. It follows from Theorems~\ref{thm:parity} and~\ref{thm:reach} that \THRESHBUD is linearly reducible to the problem of solving a stochastic reachability game, which is known to be in NP and coNP \cite{Con92}. Thus, we have the following.

\begin{theorem}
For parity bidding games, \THRESHBUD is in NP and coNP.
\end{theorem}

%\subsection*{Prompt winning}
We conclude this section by studying a stronger notion of winning that is called {\em promptness} \cite{KPV09}. %In reachability games, we say that a player wins promptly if there is a bound $k \in \Nat$ such that the target is reached by round $k$. 
{\em B\"uchi} games are a special case of parity games in which the maximal parity index is $1$: vertices with parity $1$ are called {\em accepting} and a play is winning for \PO iff it visits an accepting state infinitely often. We show a negative result for B\"uchi bidding games: intuitively, we show that under mild assumptions \PO cannot win promptly.

\stam{
Similarly, for B\"uchi games, we ask whether there is a bound $k$ such that an accepting state is visited every $k$ rounds (rather than the weaker property of visiting it infinitely often). We study promptness in bidding games, and we start with reachability games. The following theorem is an immediate corollary of the proof of Theorem~\ref{thm:richman}.

\begin{theorem}
Let $t \in \Nat$ be the minimal index such that $B_R > R(t, v)$, for some $v \in V$. Then, the reachability player can guarantee reaching $v_R$ in $t$ rounds and the safety player can keep the game from reaching $v_R$ for $t-1$ rounds.
\end{theorem}
} 

\begin{theorem}
\label{thm:prompt parity}
Consider a strongly-connected B\"uchi bidding game $\G = \zug{V,E}$ and let $F \subseteq V$ be a set of accepting vertices. If $\G$ contains a cycle that does not traverse a vertex in $F$, then for every $k \in \Nat$ and every  initial positive budget, \PT can force the game not to visit an accepting vertex for at least $k$ turns.
\end{theorem}
\begin{proof}
Let $C$ be a cycle in $\G$ with no accepting state. We construct a reachability game $Cyc(\G, C, k)=\zug{V \times \set{0,\ldots,k}, E', t'}$ from $\G$ in which we associate \PT in $\G$ with \PO in $Cyc(\G, C, k)$ and his goal is to traverse the cycle $C$ $k$ times in a row. Intuitively, the structure of $Cyc(\G, C, k)$ can be thought of as maintaining a counter: when the token is on the vertex $\zug{v,i}$, it means that $C$ was traversed for $i$ times. We describe $E'$ formally. Let $C = e_1, \ldots, e_n$ be a sequence of edges and for $e \in E$. Let $i \in \set{1,\ldots, k}$ and $e = \zug{v, u} \in E$. Suppose $e$ appears in $C$ but it is not the first edge, thus $e \neq e_1$. Then, we have $\zug{\zug{v, i}, \zug{u,i}} \in E'$, which means that when the game proceeds on an edge in $C$, the counter stays unchanged. When $e = e_1$ is the first edge in $C$ and $i<k$, we increment the counter, thus $\zug{\zug{v, i}, \zug{u, i+1}} \in E'$. When $e$ is not in $C$, we reset the counter and drop to the first level, thus $\zug{\zug{v,i}, \zug{u, 0}} \in E'$. Let $v_0$ be the first vertex in $C$. Then, the target of \PT is $\zug{v_0, k}$, which means that the cycle is traversed $k$ times in a row. 

It is not hard to see that a winning play for \PO in $Cyc(\G, C, k)$ corresponds to traversing the cycle $C$ $k$ times in a row in $\G$. Moreover, since $\G$ is strongly-connected, the target is reachable from all the vertices in $Cyc(\G, C, k)$. Thus, by Proposition~\ref{prop:Richman-Reach}, the threshold budgets are $0$ in all vertices. 

We describe a \PT strategy in $\G$ that ensures that there is no bound on the frequencies of visits to accepting states. Suppose \PT starts with a budget of $\epsilon >0$ in $\G$. He splits his budget into infinitely many parts $\epsilon_1,\epsilon_2,\ldots$. For $i \geq 1$, suppose the token is on $v \in V$. \PT plays according to \PO's winning strategy from $\zug{v, 0}$ in $Cyc(\G, C, i)$ with an initial budget of $\epsilon_i$ to force the game to cycle $C$ $i$ times in a row. Thus, for every $i \geq 1$, there is a sequence of $i \cdot |C|$ with no visit to an accepting state, and we are done.
\end{proof}

\section{Mean-Payoff Bidding Games}
\label{sec:MP}
This section consists of our most technically challenging contribution. We show that threshold budgets exist in mean-payoff bidding games and construct optimal strategies for the players. The key component of the proof is a quantitative solution to strongly-connected mean-payoff bidding games. Similar to the proof structure for parity games, the solution allows us to solve general games by first reasoning about the bottom strongly-connected components of the game and then constructing a reachability game for the rest of the vertices. 

Consider a strongly-connected mean-payoff bidding game $\G$. Recall that a play in a mean-payoff game has a payoff, which is \Min's cost and \Max's reward. Assuming both players start with a positive initial budget, we are intuitively interested in the minimal payoff \Min can guarantee assuming his budget is $r \in (0,1)$.\footnote{We use $r$ for ``ratio'' of the total budget as is used in other bidding mechanisms in which the sum of budgets is not constant. See Section~\ref{sec:poorman}.}
Since $\G$ is strongly-connected and the definition of the payoff is prefix independent, Proposition~\ref{prop:Richman-Reach} implies that the optimal payoff does not depend on the initial vertex. Thus, it is meaningful to refer to the {\em mean-payoff value} of $\G$, which we formally define as follows.

\begin{definition}
({\bf Mean-payoff value}) Consider a strongly-connected game $\G$ and $r \in (0,1)$. The mean-payoff value of $\G$ w.r.t. $r$, denoted $\MP^r(\G)$ is a value $c \in \R$ such that 
\begin{itemize}
\item If \Min's budget is greater than $r$, then he can guarantee that the payoff is at most $c$.
\item \Min cannot do better: for every $\epsilon > 0$, if \Max's initial budget is greater than $1-r$, then he can guarantee a payoff of at least $c-\epsilon$.
\end{itemize}
\end{definition}

We justify the asymmetry in the definition by noting that the definition of the payoff of a play uses $\lim\inf$ and thus gives \Min an advantage.

The following theorem consists of the main technical contribution of this section. It intuitively states that the initial budget does not matter in strongly-connected mean-payoff  bidding games and shows an extended probabilistic connection for these games. In Section~\ref{sec:poorman}, we contrast this property of the Richman-bidding mechanism that we use with the properties of mean-payoff bidding games with other bidding mechanisms. Recall that the mean-payoff value of a vertex in a stochastic mean-payoff game is the expected payoff of the game when both players play optimally. It is not hard to show that since $\G$ is strongly-connected, the mean-payoff values of all the vertices in $\RT(\G)$ is the same, thus it is meaningful to refer to the mean-payoff value of $\RT(\G)$, which we denote by $\MP(\RT(\G))$. 

\begin{theorem}
\label{thm:MP-SCC}
Consider a strongly-connected mean-payoff bidding game $\G$. The mean-payoff value of $\G$ exists and does not depend on the initial budget: there exists $c \in \R$ such that for every $r \in (0,1)$, we have $\MP^r(\G) = c$. Moreover, the value of $\G$ equals the mean-payoff value of the random-turn mean-payoff game $\RT(\G)$ in which in each turn, the player who chooses a move is selected uniformly at random, thus for every $r \in [0,1]$, we have $\MP^r(\G) = \MP(\RT(\G))$. 
\end{theorem}

The two cases of Theorem~\ref{thm:MP-SCC} are proven separately for \Min in Theorem~\ref{thm:MP-min} and for \Max in Theorem~\ref{thm:MP-max} in the following sections. We now describe the theorem's implications. Recall that \Min wins a mean-payoff game if he can guarantee that the payoff is non-positive.
\begin{theorem}
\label{thm:general-MP}
Threshold budgets exist in mean-payoff bidding games. The \THRESHBUD problem for mean-payoff bidding games is in NP and coNP.
\end{theorem}
\begin{proof}
Consider a general mean-payoff bidding game $\G = \zug{V, E, w}$. Consider $v \in V$ that belongs to a BSCC $S$ of $\G$. Let $\G_S$ be the game restricted to $S$. Theorem~\ref{thm:MP-SCC} states that if $\MP(\RT(\G_S)) \leq 0$, then with every positive initial budget, \Min can guarantee a payoff of at most $0$. Thus, the threshold budget in $v$ is $0$. On the other hand, the theorem implies that if $\MP(\RT(\G_S)) > 0$, \Max can guarantee a positive payoff with any positive initial budget, thus the threshold in $v$ is $1$. In the first case, we call $S$ winning for \Min, and in the second case, we call $S$ losing for \Min. We construct a double-reachability game $\text{DR}(\G)$ in which we associate \Min with \PO and set his target to be the set of vertices from which there is no path to losing BSCC, and the target for \Max, which we associate with \PT, is the set of vertices from which there is no path to a BSCC that is winning for \Min. Similarly to the proof of Theorem~\ref{thm:parity}, the threshold budgets in $\text{DR}(\G)$ coincide with the threshold budgets in $\G$. 

Finally, we show that \THRESHBUD is in NP by showing how to verify that $\thresh(v) \geq 1/2$. For each BSCC $S$ in $\G$, we guess positional strategies in the stochastic game $\RT(\G_S)$. In addition, we guess two target sets of vertices $T_1, T_2 \subseteq V$ and construct the reachability stochastic game  $\RT(\text{DR}(\G))$ using them. Finally, we guess two positional strategies in $\RT(\text{DR}(\G))$. We first verify that the strategies are optimal in the mean-payoff stochastic games, which can be done in polynomial time. Thus, we obtain the values in all these games. We use the values to verify our guess of the targets $T_1$ and $T_2$. Namely, we check whether every BSCC $S$ that is winning for \Min is contained in $T_1$ and that there is no path from a vertex in $T_1$ to a BSCC that is winning for \Max, and dually for $T_2$. Finally, we verify that the positional strategies in the reachability stochastic game are optimal. The solution to $\RT(\text{DR}(\G))$ gives us the threshold budget in $v$ and we accept if it is at least $1/2$. The size of the witness is polynomial in the input and the verification of the guess can be done in polynomial time. The algorithm above shows that \THRESHBUD is in coNP since the only change is to accept when $\thresh(v) < 1/2$.
\end{proof}

\subsection{An optimal \Min strategy in strongly-connected mean-payoff bidding games}
\label{sec:Min}
In this section we construct an optimal strategy for \Min in a strongly-connected mean-payoff bidding game. Since the definition of payoff favors \Min, this is technically easier than the construction for an optimal strategy for \Max, which we construct in the following section. 

Consider a strongly-connected mean-payoff bidding game $\G$. In this section, we assume w.l.o.g. that $\MP(\RT(\G)) = 0$ as otherwise we can decrease all weights by this value. We construct a bidding strategy for \Min that, with any positive initial budget, guarantees that the payoff is non-positive.

Recall that the energy of a finite play is the sum of the weights that it traverses. The following lemma shows that it suffices to construct a \Min strategy that keeps the energy bounded from above. 

\begin{lemma}
\label{lem:simplify-Min}
Consider a mean-payoff bidding game $\G$. Suppose that for every positive initial budget $\epsilon > 0$ and initial energy $k_I \in \Nat$, there is a constant $N \in \Nat$ such that \Min has a strategy $f_m$ that keeps the energy bounded by $N$. That is, for every \Max strategy $f_M$ and initial vertex $u$, a finite play $\pi = \play(u, f_m, f_M)$ either reaches energy $0$ or has $E(\pi^n)  \leq N$, for every $1 \leq n \leq |\pi|$. Then, \Min can guarantee a non-positive payoff in $\G$.
\end{lemma}
\begin{proof}
Suppose \Min has a strategy $f_m$ as the above, and we describe a \Min strategy $f'_m$ that guarantees a non-positive payoff. Suppose \Min's initial budget is $\epsilon > 0$. He splits his budget into infinitely many parts $\epsilon_1,\epsilon_2,\ldots$. Initially, \Min plays according to $f_m$ as if his budget is $\epsilon_1$ until an energy of $0$ is reached. When energy $0$ is reached again, he bids $0$ until the energy increases. Once the energy is positive, \Min plays according to $f_m$ with an initial budget of $\epsilon_2$ until an energy of $0$ is reached, and so on. Thus, the strategy guarantees that either (1) an energy of $0$ is reached infinitely often, or (2) if at some point an energy of $0$ is never reached, then the energy stays bounded from above. Recall that the definition of the payoff of an infinite path $\eta = v_1,v_2,\ldots$ is $\payoff(\eta) = \lim\inf_{n \to \infty} E(\eta^n)/n$. Note that an infinite path that satisfies one of the properties (1) or (2) above has a non-negative payoff.
\end{proof}

\noindent{\bf The importance of moving in a vertex.}
The first component of the strategy construction devises a measure of how ``important'' it is to move in each vertex in the game. Our definition relies on the concept of {\em potential}, which was defined in the context of the strategy improvement algorithm to solve stochastic games \cite{How60}.  The potential of $v$, denoted $\Po(v)$, is a known concept in probabilistic models and its existence is guaranteed \cite{Put05}. We formalize the notion of the ``importance'' of moving in a vertex $v$ by defining its {\em strength}, which we denote by $\St(v)$, and is formally the maximal difference in potentials of the neighbors of $v$.

\begin{definition}
({\bf Potentials and strengths})
Consider two optimal positional strategies $f$ and $g$ in $\RT(\G)$, for \Min and \Max, respectively. Recall that when constructing $\RT(\G)$, for every vertex $v \in V$, we add two copies $v_{\Min}$ and $v_{\Max}$, that are controlled by \Min and \Max, respectively. For $v \in V$, let $v^-,v^+ \in V$ be such that $f(v_\Min) = v^-$ and $g(v_\Max) = v^+$. The potential of $v$ is a function that satisfies the following and the strength in $v$ is the difference in potentials:
\[
\Po(v) = \frac{1}{2}\big( \Po(v^+) + \Po(v^-)\big)  + w(v)  - \MP(\RT(\G))
\text{ and }
\St(v) = \frac{1}{2} \big(\Po(v^+) - \Po(v^-)\big)
\]
\end{definition}

There are optimal strategies for which $\Po(v^-) \leq \Po(v') \leq \Po(v^+)$, for every $v' \in N(v)$, which can be found for example using the strategy iteration algorithm.

Consider a strongly-connected mean-payoff bidding game $\G = \zug{V, E, w}$. Consider a finite path $\eta = v_1,\ldots, v_n$ in $\G$. We intuitively think of $\eta$ as a play, where for every $1 \leq i < n$, the bid of \Min in $v_i$ is $\St(v_i)$ and he moves to $v_i^-$ upon winning. Thus, if $v_{i+1} = v_i^-$, we say that \Min won in $v_i$, and if $v_{i+1} \neq v_i^-$, we say that \Min lost in $v_i$. Let $W(\eta)$ and $L(\eta)$ respectively denote the indices in which \Min wins and loses in $\eta$. We call \Min wins {\em investments} and \Min loses {\em gains}, where intuitively he {\em invests} in increasing the energy and {\em gains} budget whenever the energy decreases. Let $G(\eta)$ and $I(\eta)$ be the sum of gains and investments in $\eta$, respectively, thus $G(\eta) = \sum_{i \in L(\eta)} \St(v_i)$ and $I(\eta) =  \sum_{i \in W(\eta)} \St(v_i)$. Recall that the energy of $\eta$ is $E(\eta) = \sum_{1 \leq i <n} w(v_i)$. The following lemma connects the strength, potential, and accumulated energy.%, and its proof can be found in App.~\ref{app:magic}.

\begin{lemma}
\label{lem:magic}
Consider a strongly-connected game $\G$ with $\MP(\RT(\G))= 0$, and a finite path $\eta$ in $\G$ from $v$ to $u$. Then, $\Po(v) - \Po(u) \geq E(\eta) -  G(\eta) + I(\eta)$. 
\end{lemma}
\begin{proof}
We prove by induction on the length of $\eta$. For $n=1$, the claim is trivial since both sides of the equation are $0$. Suppose the claim is true for paths of length $n$ and we prove for paths of length $n+1$. We distinguish between two cases. In the first case, \Min wins in $v$, thus the second vertex in $\eta$ is $v^-$. Let $\eta'$ be the prefix of $\eta$ starting from $v^-$. Note that since \Min wins the first bidding, we have $G(\eta) = G(\eta')$ and $I(\eta) = \St(v) + I(\eta')$. Also, we have $E(\eta) = E(\eta') + w(v)$. Combining these, we have $E(\eta) - G(\eta) + I(\eta) = E'(\eta) + w(v) - G(\eta') + I(\eta') + \St(v)$. By the induction hypothesis, we have $\Po(v^-) - \Po(u) \geq E(\eta') - G(\eta') + I(\eta')$. Combining these with the definition of $\St(v)$, we have the following. 
\[E(\eta) - G(\eta) + I(\eta) \leq  \St(v) + \Po(v^-) +w(v) - \Po(u) = \]\[=\frac{1}{2}\big(\Po(v^+)-\Po(v^-)\big) + \Po(v^-)+w(v)- \Po(u)= \Po(v) - \Po(u).\]

We continue to the second case in which \Max wins in $v$ and let $v'$ be the second vertex in $\eta$. Recall that we have $\Po(v^+) \geq \Po(v')$. Dually to the first case, we have $G(\eta) = \St(v) + G(\eta')$ and $I(\eta) = I(\eta')$. 
\[ 
E(\eta) - G(\eta) + I(\eta) = E(\eta') - G(\eta') + I(\eta') - \St(v) + w(v) \leq\]
\[\leq \Po(v') - \St(v) + w(v) - \Po(u) \leq \Po(v^+) - \frac{1}{2}\big(\Po(v^+)- \Po(v^-)\big) + w(v) - \Po(u) = \Po(v) - \Po(u).
\]
\end{proof}

\begin{figure}[ht]
\centering
\includegraphics[height=2cm]{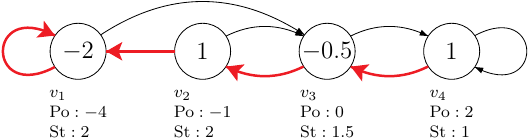}
\caption{A mean-payoff bidding game $\G$ with $\MP(\RT(\G)) = 0$, where we depict the weight of a vertex inside it, its potential and strength below it, and with a bold edge the vertex to which \Min moves the token upon winning a bidding.}
\label{fig:MP-example}
\end{figure}

\begin{example}
Consider the game that is depicted in Figure~\ref{fig:MP-example}. For each vertex $v$, we depict $v^-$ by using a bold edge. We illustrate Lemma~\ref{lem:magic}. Consider the path $\eta = v_3, v_2, v_1, v_1, v_4$, which intuitively corresponds to \Min winning three biddings and then losing two. For example, we have $v^-_3 = v_2$ and $v^-_1 = v_1$, thus when losing the bidding in $v_1$, \Max would proceed to $v_3$. The energy is $E(\eta) = -0.5+1-2-2-0.5 = -4$ (recall that the last vertex does not contribute to the accumulated energy), the gain is $G(\eta) = \St(v_1) + \St(v_4) = 3.5$, the investment is $I(\eta) = \St(v_3) + \St(v_2) + \St(v_1) = 5.5$, and the potentials of the two end points are $\Po(v_3) = 0$ and $\Po(v_4) = 2$. Plugging in the values, we have $\Po(v_3) - \Po(v_4) = 0-2 = -4 -3.5 + 5.5 = E(\eta) - G(\eta) + I(\eta)$.\hfill\qed
\end{example}

\noindent{\bf Normalizing the strengths.}
The second component in \Min's strategy is a normalization of the strengths, which guarantees that the energy is bounded from above. We develop the intuition in the following example.

\begin{example}
\label{ex:invariant}
Consider the  mean-payoff bidding game that is depicted in Figure~\ref{fig:loops}. In Example~\ref{ex:tit-for-tat}, we showed the tit-for-tat \Min strategy that bounds the energy from above, which, by Lemma~\ref{lem:simplify-Min}, suffices to guarantee a non-positive payoff. We show an alternative construction that generalizes to general strongly-connected games. 

Again, \Min always proceeds to $v_2$ upon winning a bidding and we assume \Max proceeds to $v_1$ upon winning a bidding. The difficulty is finding the right bids. It is convenient to assume that the initial energy is a positive number $k_I \in \Nat$ rather than $0$. Suppose that \Min starts with a positive initial budget of $B^{I}_m > 0$. \Min chooses an $N \in \Nat$ such that $B^{I}_m > \frac{k_I}{N}$, which is clearly possible since $k_I$ is a constant and $B_I$ is positive. \Min always bids $\frac{1}{N}$ as long as the energy is positive.

We show that the following invariant is maintained: if the energy level reaches $0 \leq k \in \Nat$, \Min's budget is greater than $\frac{k}{N}$. First, our choice of $N$ implies that the invariant holds initially. Second, assuming that the invariant holds before a bidding, we show that it holds after it. Suppose that the energy is $k$ and \Min's budget is $k/N + \epsilon$. If \Min wins, the energy decreases by $1$ to $k-1$ and his budget decreases by $1/N$ to $(k-1)/N + \epsilon$. On the other hand, if \Max wins the bidding, he bids at least as much as \Min, thus \Min's budget increases by at least $1/N$. The energy increases by $1$ to $k+1$ and \Min's budget increases to $(k+1)/N + \epsilon$. The invariant implies that if the energy does not reach $0$, then it is bounded by $N$. Indeed, if the energy reaches $k=N$, \Min's budget is $N/N + \epsilon$, which is impossible since the sum of budgets is $1$.\qed
\end{example}

%Extending this result to general strongly connected games is not immediate.  Consider a strongly-connected game $\G = \zug{V, E, w}$ and a vertex $u \in V$. We would like to maintain the invariant that upon reaching $u$ with energy $k$, the budget of \Min exceeds $k/N$, for a carefully chosen $N$. The game in the simple example above has two favorable properties that general SCCs do not necessarily have. First, unlike the game in the example, there can be infinite paths that avoid $u$, thus \Min might need to invest budget in drawing the game back to $u$. Moreover, different paths from $u$ to itself may have different energy levels, so bidding a uniform value (like the $\frac{1}{N}$ above) is not possible. 

We describe the intuition behind \Min's strategy. In the example above, \Min's strategy puts a price of $1/N$ on changing the energy: whenever the energy decreases by $1$, he pays $1/N$, and whenever the energy increases, he gains at least $1/N$. Lemma~\ref{lem:magic} allows us to generalize this connection between changes in energy and changes in budget. In a vertex $v$, \Min bids $\frac{1}{N} \cdot \St(v)$ and proceeds to $v^-$ upon winning. For example, consider the game that is depicted in Figure~\ref{fig:MP-example} and the cycle $\pi = v_3, v_2, v_1, v_3$ that results from \Min winning the two biddings followed by a \Max win. The change in energy is $w(v_3) +w(v_2) + w(v_1) = -0.5+1-2=-1.5$ and \Min's budget decreased by at most $\frac{1}{N} \big( \St(v_3) + \St(v_2) - \St(v_1)\big) = (1.5+2-2)/N=1.5/N$. Thus, \Min invests at most $c/N$ in a decrease of $c$ units of energy, and he gains at least $c/N$ units of budget when the energy increases by $c$ units. A similar argument as in the lemma above shows an invariant between the energy and budget and in turn, that the energy stays bounded from above.

To formally define \Min's strategy we show how to choose $N$ in general graphs, which requires some book-keeping due to paths that are not cycles. We call \Min's strategy $f_m$. Consider a positive initial budget $B \in (0,1]$ for \Min and an initial energy $k_I \in \Nat$. Let $\Po_M=\max_{v \in V} |\Po(v)|$ and $\St_M = \max_{v \in V} |\St(v)|$. We choose $N \in \Nat$ such that $B > \frac{k_I + \St_M + 2\Po_M}{N}$. When the game reaches $v \in V$, \Min bids $\St(v)/N$ and moves to $v^-$ upon winning. %We formalize the intuition of tying energy and budget by means of an invariant that is maintained throughout a play. 

\begin{lemma}
\label{lem:Min-invariant}
Consider a \Max strategy $f_M$, an initial energy $k_I \in \Nat$, and let $\pi = \play(f_m, f_M)$ be a finite play whose energy stays positive. Thus, for every prefix $\pi^n$, for $0 \leq n \leq |\pi|$, we have $k_I + E(\pi^n) >0$. Let $k = k_I + E(\pi)$ be the energy following $\pi$. Then, \Min's budget following $\pi$ is at least $\frac{k+\St_M}{N}$.
\end{lemma}
\begin{proof}
The invariant clearly holds initially. With a slight abuse of notation, let $G(\pi)$ be the sum of ``gains'' in $\path(\pi)$, namely the sum of strengths in vertices in which \Max wins the bidding, and similarly $I(\pi)$ be the ``investments'' in $\path(\pi)$, namely the sum of strengths in vertices in which \Min wins the bidding. Let $B$ be \Min's initial budget and $B'$ his budget following $\pi$. Since \Min bids $\St(v)/N$ in a vertex $v$, we have $B' = B + \big(G(\pi) - I(\pi)\big)/N$. From Lemma~\ref{lem:magic}, we have $2\Po_M -E(\pi) \geq I(\pi)-G(\pi)$. By combining with $k = k_I + E(\pi)$ and re-arranging, we have 
\[
2\Po_M -E(\pi) \geq  I(\pi)-G(\pi) = N(B- B')
\]
\[
B- \frac{2\Po_M- k+k_I}{N} \leq B'
\]
Since we define $B > \frac{k_I + \St_M + 2\Po_M}{N}$, we obtain that $B' > \frac{k+\St_M}{N}$, and we are done.
\end{proof}

Note that the strategy $f_m$ is legal, i.e., \Min always has sufficient budget to bid according to $f_m$. Indeed, for a choice $N \in \Nat$ made by the strategy, the maximal bid in a vertex in $\G$ is $\St_M/N$. Lemma~\ref{lem:Min-invariant} implies that \Min has sufficient budget for the bid. Moreover, since \Min's budget cannot exceed $1$, Lemma~\ref{lem:Min-invariant} implies that if the energy does not reach $0$, then it is bounded by $N-\St_M$. Combining with Lemma~\ref{lem:simplify-Min}, we obtain the first direction in Theorem~\ref{thm:MP-SCC}.

\begin{theorem}
\label{thm:MP-min}
Let $\G$ be a strongly-connected mean-payoff bidding game with $\MP(\RT(\G))=0$. Then, from every vertex in $\G$ and with any positive initial budget, \Min can guarantee a non-positive payoff.
\end{theorem}

\subsection{An optimal \Max strategy in strongly-connected mean-payoff bidding games}
In this section we focus on the more challenging task of constructing an optimal strategy for \Max: Given a strongly-connected mean-payoff bidding game $\G$ with $\MP(\RT(\G)) > 0$, we construct a bidding strategy for \Max in $\G$ that guarantees a positive payoff. 

The following lemma reduces the problem of optimizing the payoff to the problem of bounding the energy from below.

\begin{lemma}
\label{lem:simplify-max}
Assume that for every \Max initial budget $\epsilon > 0$ in a game $\G$ with $\MP(\RT(\G)) > 0$, he can keep the energy bounded from below by a constant $N(\G, \epsilon) \in \Z$. Then, \Max can guarantee a positive mean-payoff value in $\G$.
\end{lemma}
\begin{proof}
Let $\G'$ be a mean-payoff bidding game that is obtained from $\G$ by decreasing $\MP(\RT(\G))/2$ from all the weights in $\G$. It is not hard to see that $\MP(\RT(\G'))  = \MP(\RT(\G))/2$ and in particular it is positive. Let $\epsilon > 0$, and suppose \Max plays in $\G$ according to a strategy that keeps the energy above $N(\G', \epsilon)$ in $\G'$. For a finite play $\pi$ in $\G$ we have $E(\pi) \geq |\pi| \cdot \MP(\RT(\G))/2 + N(\G', \epsilon)$. Since $N(\G', \epsilon)$ is a constant, its contribution to the payoff vanishes as the length of $\pi$ tends to infinity, thus the payoff is at least $\MP(\RT(\G))/2$, which is positive.
\end{proof}

Bounding the energy from below is more challenging than \Min's goal in the previous section of bounding the energy from above. A first attempt for constructing a \Max strategy would be to use a similar strategy as the previous section only with reversed roles: \Max's strategy would guarantee that whenever the energy is $k \in \Nat$, his budget exceeds $k/N$, for some $N \in \Nat$. He would ensure that whenever the energy increases by one unit, his budget decreases by at most $1/N$, and whenever the energy decreases by one unit, his budget increases by at least $1/N$. This attempt fails since \Min reacts by allowing \Max to win for a while and draw the energy all the way up to $N$, where \Max's budget runs out. When \Min has all (or most) of the budget, he can win an arbitrary number of biddings in a row. Thus, he can draw the energy arbitrary low, causing \Max to lose since the energy would not be bounded from below. The moral of this attempt is that \Max should avoid exhausting his budget. He cannot use a fixed normalization factor of $1/N$. Rather, the normalization factor should decrease as the energy increases. In the next two sections we devise a normalization scheme, first in simpler strongly-connected components and then in general ones.

\subsubsection{An optimal \Max strategy in recurrent mean-payoff bidding games}
A game $\G$ is called {\em recurrent}, if it is strongly-connected and there is a vertex $u \in V$ such that every cycle in $\G$ includes $u$ (see Figure~\ref{fig:recurrent}). We refer to $u$ as the {\em root} of $G$. In this section, we construct an optimal strategy for \Max in a recurrent mean-payoff bidding game.

\noindent{\bf An adapted definition of importance.}
Recall that \Min's strategy in the previous section matches changes in budget with changes in energy. The first component in \Max's strategy makes this connection asymmetric: we find $z > 1$, such that when the energy increases by $c$ units, \Max invests at most $c$ units of budget, but when the energy decreases by $c$ units, \Max gains at least $z \cdot c$ units of budget.

Consider a recurrent mean-payoff bidding game $\G = \zug{V, E, w}$ with $\MP(\RT(\G)) > 0$. We alter the weights to give advantage to \Min. For $z >1$, let $\G^z = \zug{V, E, w^z}$, where
\[
w^z(v) = \begin{cases} w(v) & \text{ if } w(z) \geq 0 \\
z \cdot w(v) & \text{ if } w(z) < 0
\end{cases}
\]
Clearly, $\MP(\RT(\G)) \geq \MP(\RT(\G^z))$. We select $z > 1$ such that  $\MP(\RT(\G^z)) \geq 0$. This is possible since by additively changing all the weights in $\RT(\G)$ by a constant $c$, the value changes by $c$. We select $z$ such that $z \cdot \max_{v \in V} w(v) \leq \MP(\RT(\G))$. 

Consider a finite path $\eta$ in $\G$. The following lemma connects the energy of $\eta$ in $\G$ with its energy in $\G^z$. Note that $E(\eta)$ might be negative thus neither claim follows from the other. Let $E^z(\eta)$ be the accumulated energy in $\G^z$.

\begin{lemma}
\label{lem:E-Ez}
Consider a finite path $\eta \in cycles(u)$. Then, $E(\eta) \geq E^z(\eta)$ and $zE(\eta) \geq E^z(\eta)$. 
\end{lemma}
\begin{proof}
Let $E^{\geq 0}(\eta)$ and $E^{<0}(\eta)$ be the sum of non-negative weights and negative weights in $\eta$, respectively. We have $E(\eta) = E^{\geq 0}(\eta)+E^{<0}(\eta)$ and $E^z(\eta) = E^{\geq 0}(\eta)+zE^{<0}(\eta)$. The inequality $E(\eta) \geq E^z(\eta)$ is immediate. For the second inequality, we multiply the first equality by $z$ and subtract it from the first to get $E^z(\eta) - zE(\eta)= E^{\geq 0}(\eta) - zE^{\geq 0}(\eta)\leq 0$, and we are done.
\end{proof}

We adapt Lemma~\ref{lem:magic} to our setting. We find optimal positional strategies $g_m$ and $g_M$ for \Min and \Max, respectively, in the stochastic game $\RT(\G^z)$. Using them, we define, for each vertex $v \in V$, vertices $v^-$ and $v^+$ by setting $v^- = g_m(v_m)$ and $v^+ = g_M(v_M)$. We respectively denote by $\Po^z$ and $\St^z$, the potential and strength functions of $\G^z$. For a finite path $\eta = v_1,\ldots,v_n$, we denote  $G^z(\eta) = \sum_{v_{i+1} \neq v_i^+} \St^z(v_i)$ and $I^z(\eta) = \sum_{v_{i+1} = v_i^+} \St^z(v_i)$. The proof of the following lemma is dual to the proof of Lemma~\ref{lem:magic}.

\begin{lemma}
\label{lem:magic-max}
For a finite path $\eta$ from $v$ to $u$, we have $\Po^z(v) - \Po^z(u) \leq E^z(\eta) + G^z(\eta) - I^z(\eta)$.
\end{lemma}

Let $pay(\eta) = I^z(\eta) - G^z(\eta)$. Combining the two lemmas above, we obtain the required asymmetry between ``gaining'' and ``investing''.

\begin{lemma}
\label{lem:max-budget}
Consider a path $\eta \in cycles(u)$. When $E(\eta) \geq 0$, we have $pay(\eta) \leq E(\eta)$, and when $E(\eta) < 0$, we have $-pay(\eta) \geq -z \cdot E(\eta)$.
\end{lemma}

\begin{figure}[ht]
\centering
\includegraphics[height=3cm]{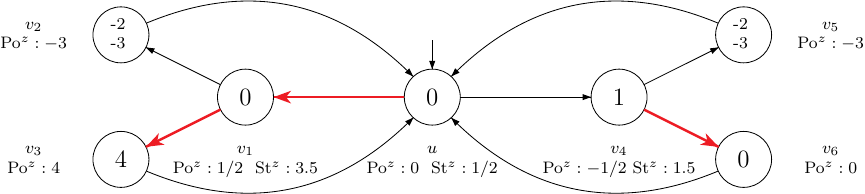}
\caption{A recurrent mean-payoff bidding game $\G$ with $\MP(\RT(\G)) > 0$. For $z=4/3$, the alternated weights in a vertex are depicted below the original weight.}
\label{fig:recurrent}
\end{figure}

\begin{example}
\label{ex:recurrent}
Consider the recurrent mean-payoff bidding game $\G$ that is depicted in Figure~\ref{fig:recurrent}. We have $\MP(\RT(\G)) = 0.5$, thus \Max can guarantee a positive payoff. We illustrate Lemma~\ref{lem:max-budget}. The weights of the vertices in $\G$ are depicted on top, and, for $z=3/2$, the weights of the negative-weighted vertices are depicted below. Consider the path $\eta = u, v_1, v_2, u$. Thus, \Max wins the first bidding and loses the second. The bids in vertices with only one outgoing edge are $0$. With this choice of $z$, we have $\MP(\RT(\G^z))=0$, thus we get equality between energy and budget in $\G^z$. Indeed, the change in energy in $\G^z$ is $E^z(\eta) = w^z(u) + w^z(v_1) + w^z(v_2) = -3$. On the other hand, \Max's gain in $\eta$ is $\St^z(v_1) = 3.5$ and his investment is $\St^z(u) = 0.5$, thus his budget increases by $3.5-0.5 = 3$. However, the ``real'' change in energy is the one exhibited in $\G$, which is $E(\eta) = w(u) + w(v_1) + w(v_2) = -2$. Thus, in a decrease of $2$ units of energy, \Max gains $2\cdot z = 3$ units of budget rather than only $2$. 

The worst case for \Max is in paths that traverse only positive or only negative weights. In paths that traverse a mix of weights, the inequality in Lemma~\ref{lem:max-budget} is strict. For example, consider the path $u, v_4, v_5,u$. The change in energy is $w(u) + w(v_4) + w(v_5) = -1$ and the change in budget is $\St^z(u) + \St^z(v_4) = 0.5+1.5 = 2 > 1 \cdot 3/2$.\qed
\end{example}

%We describe the intuition for the normalization scheme that we construct. We select $N \in \Nat$ and partition the natural numbers into {\em energy blocks} of size $N$. For $n \geq 1$, we associate a normalization factor $\gamma_n \in [0,1]$ with the $n$-th block, which we call the {\em currency} of the block. Whenever the energy is in the $n$-th block, \Max bids according to the currency of the block while matching energy and budget. We define the $\gamma_n$'s to be decreasing so that \Max's bid decreases as the energy increases. The easy case is when the energy stays within a block, since that is similar to the previous section. The difficulty is handling plays that switch between blocks. For example, consider a play with a sinusoidal energy behavior: \Min alternates between winning and losing and the energy increases by, say $c \in \Nat$ units, then decreases by $c$, increases by $c$, and so forth. Moreover, the wave ``sits'' on top of the $n$-th block. So, each time the energy increases, \Max ``invests'' $c$ units of budget in the currency $\gamma_n$, and whenever \Max ``gains'' $c$ units of budget, it is in the lower currency $\gamma_{n+1}$. This would cause \Max's budget to eventually run out. To overcome this issue, we develop the idea of tying energy and budget to give an advantage to \Max: investing in the $n$-block is done in the currency $\gamma_n$ while gaining is done in the higher currency $\gamma_{n-1}$ of the block below.

\noindent{\bf \Max's strategy.}
When the game reaches a vertex $v$, \Max bids $\St(v) \cdot \gamma$, where $\gamma$ is a normalization factor that depends on the energy in the last visit to $u$. That is, the normalization changes only after visiting $u$. In order to define $\gamma$, we select $N \in \Nat$ and partition the natural numbers into {\em energy blocks} of size $N$. Each energy block is associated with its own normalization, which we call the {\em currency} of the block. Recall that we chose $z > 1$. For $n \in \Nat$, the currency of the $n$-th block is $z^{-n}$. The key idea follows from combining with Lemma~\ref{lem:max-budget}: investing in the $n$-th block is done in the currency of the $n$-th block while gaining in the $n$-block is in the higher currency of the $(n-1)$-th block.

\begin{example}
Consider the game $\G$ that is depicted in Figure~\ref{fig:recurrent}, and consider two plays $\pi_1 = u, v_1,v_3, u$ and $\pi_2 = u, v_1, v_2, u$. Suppose the energy in $u$ is in the $3$-rd block, thus the currency is $1.5^{-3}$. In $\pi_1$, the energy increases, i.e., we have $E(\pi_1) = 4$, and \Max matches his change in budget in the currency of the $3$-rd block, i.e., \Max invests $(0.5 + 3.5)\cdot 1.5^{-3} = 4 \cdot 1.5^{-3}$. On the other hand, in $\pi_2$, we have a decrease of energy, i.e., we have $E(\pi_2) = -2$, and \Max gains $(-0.5 + 3.5) \cdot 1.5^{-3} = 2 \cdot 1.5^{-2}$, thus we have a connection between changes in energy and budget, only in the higher currency of the $2$-nd block.\qed
\end{example}

We choose $N \in \Nat$ as follows. Let $cycles(u)$ be the set of paths that are simple cycles from $u$ to itself. A crucial advantage of recurrent games is that all cycles pass through $u$. Our definition relies on the maximal energy of such a cycle, which we denote by $E_M = \max_{\eta \in cycles(u)} |E(\eta)|$. We choose $N \in \Nat$ such that $N \geq (\St^z_M + 3E_M)/(1-z^{-1})$, where $\St^z_M$ is the maximal strength of a vertex in $\G^z$. For $n \geq 1$, we refer to the $n$-th block as $N_n$, and we have $N_n = \set{N (n-1), N(n-1)+1,\ldots, Nn -1}$.  We use $\beta^\downarrow_n$ and $\beta^\uparrow_n$ to mark the upper and lower boundaries of $N_n$, respectively. We use a  $N_{\geq n}$ to denote the set $\set{N_{n}, N_{n+1}, \ldots}$. Consider a finite play $\pi$ that ends in $u$ and let $visit_u(\pi)$ be the set of indices in which $\pi$ visits $u$. Let $k_I \in \Nat$ be an initial energy. We say that $\pi$ {\em visits} $N_n$ if $k_I + E(\pi) \in N_n$. We say that $\pi$ {\em stays in}  $N_n$ starting from an index $1 \leq i \leq |\pi|$ if for all $j \in visit_u(\pi)$ such that $j \geq i$, we have $k_I + E(\pi_1,\ldots,\pi_j) \in N_n$.

We are ready to describe \Max's strategy, which we denote by $f_M$. \Max chooses $k_I \in \Nat$ and plays as if the initial energy in $k_I$. With the right choice of $k_I$, his strategy will keep the energy non-negative. In turn, assuming that the real initial energy is $0$, we obtain that the energy stays above $-k_I$. Suppose the game reaches a vertex $v$ and the energy in the last visit to $u$ was in $N_n$, for $n \geq 1$. Then, \Max bids $z^{-n} \cdot \St^z(v)$ and proceeds to $v^+$ upon winning.
Consider an initial \Max budget $B^I_M >0$. We choose an initial energy $k_I \in \Nat$ with which $f_M$ guarantees that energy level $0$ is never reached. Recall the intuition that increasing the energy by a unit requires an investment of a unit of budget in the right currency. Thus, increasing the energy from the lower boundary $\beta_n^\downarrow$ of $N_n$ to its upper boundary $\beta_n^\uparrow$, costs $N \cdot z^{-n}$. We define $cost(N_n)=N \cdot z^{-n}$ and $cost(N_{\geq n}) = \sum_{i=n}^\infty cost(N_n)$. A first attempt for the definition of $k_I$ would be $\beta_n^\downarrow$ such that $B^I_M > cost(N_{\geq n})$, which intuitively means that \Max's initial budget would never run out even if he always wins. This is almost correct. We need some {\em wiggle room} to allow for changes in the currency. Also, note that drawing the energy to $0$ from $\beta_n^\downarrow$ would cost \Min a total of $\sum_{i=1}^n cost(N_i)$. We choose $k_I$ so that this cost is greater than $1$, thus we ensure that the energy never reaches $0$.

\begin{definition}
\label{def:k_I}
Let $k_I$ be $\beta_n^\downarrow$ such that $B^I_M > wiggle\cdot z^{-(n-1)}+ cost(N_{\geq n})$, where $wiggle = 2E_M + \St^z_M$, and $\sum_{i=1}^n cost(N_i) > 1$.
\end{definition}

\noindent{\bf Correctness.}
We prove an invariant on \Max's budget throughout the game, which will imply that the energy never reaches $0$ when it starts from $k_I$, and hence the correctness of the strategy.

Consider a \Min strategy $f_m$, and let $\pi=play(f_m, f_M)$ be a finite play. Let $visit_u(\pi) = \tau^1\cdot \ldots \cdot \tau^m$ be a partition of $\pi$ such that for each $1 \leq i \leq m$, the $\path(\tau_i)$ is a cycle-less path that ends in $u$. We define a coarser partition of $\pi$ into sub-plays in which the same currency is used (recall that we change currency at $u$ and when switching between energy blocks). Let $\pi = \pi_1 \cdot \pi_2 \cdot \ldots \cdot \pi_\ell \cdot \pi_{\ell+1}$, where for each $1 \leq i \leq \ell$, we have $\pi_i = \tau^{i_1} \cdot \ldots \tau^{i_{n_i}}$, there is an energy block $N_n$ such that the sub-play $\tau^{i_1} \cdot \ldots \cdot \tau^{i_{n_i-1}}$ stays in $N_n$ and the sub-play $\pi_i$ visits a neighboring energy block $N_{n-1}$ or $N_{n+1}$. We then call $N_n$ the energy block of $\pi_i$. We use $e^i$ to denote the energy at the end of $\pi^i$, thus $e^i = k_I + E(\pi^i)$. Let $N_n$ be the energy block of $\pi_i$. There can be two options; either the energy decreases in $\pi_i$, thus the energy before it $e^{i-1}$ is in $N_{n+1}$ and the energy after it $e^i$ is in $N_n$, or it increases, thus $e^{i-1} \in N_{n-1}$ and $e^i \in N_n$.  We then call $\pi^i$ {\em decreasing} and {\em increasing}, respectively.

Recall that $\beta^\uparrow_n$ and $\beta^\downarrow_n$ are the upper and lower boundaries of the energy block $N_n$. Further recall that $E_M$ is the largest energy of a cycle in $\G$. Thus, whenever the energy enters $N_n$ it is within $E_M$ of the boundary (see Figure~\ref{fig:cases-recurrent}). In the case that $\pi^i$ is decreasing, the energy at the end of $\pi^i$ is $e^i \geq \beta^\uparrow_n - E_M$ and in the case it is increasing, we have $e^i \leq \beta^\downarrow_n + E_M$. Let $\ell_0 = 0$, and for $i \geq 1$, let $\ell_i = (\beta^\downarrow_{n+1} - E_M) - e^i$ in the first case and $\ell_i = (\beta^\downarrow_n + E_M) - e^i$ in the second case. Note that $\ell_i \in \set{0, \ldots, 2E_M}$. We prove the following invariant on \Max's budget when changing between energy blocks.
\begin{lemma}
\label{lem:max-invariant}
For every $i \geq 0$, suppose $\pi^i$ ends in $N_n$. The budget of \Max at the end of  $\pi^i$ is at least $(wiggle + \ell_i) \cdot z^{-(\hat{n} -1)} + cost(N_{\geq \hat{n}})$, where $\hat{n} = n+1$ if $\pi^i$ is decreasing and $\hat{n} = n$ if $\pi^i$ is increasing.
\end{lemma}
\begin{proof}
The proof is by induction. The base case follows from our choice of initial energy. For $i\geq 1$,  assume the claim holds for $\pi^{i-1}$ and we prove for $\pi^i$. There are four cases for the energy changes in $\pi_i$, which we depict in Figure~\ref{fig:cases-recurrent}. Recall that $E_M$ is maximal energy of a simple cycle from $u$ and that we switch currencies at $u$. Thus, whenever we switch currency it means that the play visits a new energy block, and the first location in the block is within $E_M$ of the boundary.

Intuitively, Case~$1$ is the simplest and follows from matching energy and budget in $N_n$. In Cases~$3$ and~$4$ \Max invests and gains in the ``wrong'' currency. For example, in Case~$3$, if investing and gaining was in the same currency, \Max would have gained in the currency of $N_n$ instead of the higher currency of $N_{n-1}$. Finally, in Case~$2$, we again use this mismatch to ensure that the gain ``covers'' the cost of $N_n$ and in addition there is a ``surplus'' that covers the required wiggle room.

Let $e^{i-1}$ be the energy at the end of $\pi^{i-1}$. Consider Cases~$1$,~$3$, and~$4$ in the figure. We prove the first of these case and the others are similar. In Case~$3$, we have $e^{i-1} \in N_{n+1}$, and $\pi_i$ decreases into $N_n$ and $e^i$ is near $\beta^\uparrow_n$. Thus, we have  $\ell_{i-1} = (\beta^\downarrow_{n+1} + w_M) - e^{i-1}$ and $\ell_i = (\beta^\downarrow_{n+1} + w_M) -e^i$. Since we decrease in blocks, we have $\ell_{i-1} < \ell_i$ and  $E(\pi_i) = \ell_{i-1} - \ell_i$. By Lemma~\ref{lem:max-budget}, we have $z^{n+1} \cdot pay(\pi_i) \geq z\cdot (\ell_{i-1} - \ell_i)$, thus the gain in budget in $\pi_i$ is at least $(\ell_i - \ell_{i-1})z^{-n}$. The induction hypothesis states that \Max's budget in $\pi^{i-1}$ is at least $(E_M+\St_M + \ell_{i-1}) \cdot z^{-n} + \sum_{j=n}^\infty N  z^{-j}$, thus his budget after $\pi^i$ is at least $(E_M+\St_M + \ell_i) \cdot z^{-n} + \sum_{j=n}^\infty N  z^{-j}$, and we are done. 
The final case, which is similar to Case $2$ in the figure with a slight difference; the figure depicts energy that crosses $N_n$ and we prove for a energy that crosses $N_{n+1}$ and ends in $N_n$. That is, the energy at $\pi^{i-1}$ is in $N_{n+1}$ and $e^{i-1} \geq \beta^\uparrow_{n+1}-E_M$ and $e^i \leq \beta^\downarrow_{n+1} = \beta^\uparrow_n$. The decrease in energy is $E(\pi_i) = (2E_M - \ell_{i-1}) + (N-2E_M) + \ell_i$, thus by Lemma~\ref{lem:max-budget}, the increase in budget is $E(\pi_i) \cdot z^{n-1}$. We chose $N$ such that $(N-2E_M) \cdot z^{-(n-1)} \geq (E_M + \St_M)\cdot z^{-(n-1)} + N\cdot z^{-n}$. The claim follows from combining with the induction hypothesis, and we are done.
\begin{figure}[ht]
\centering
\includegraphics[height=5cm]{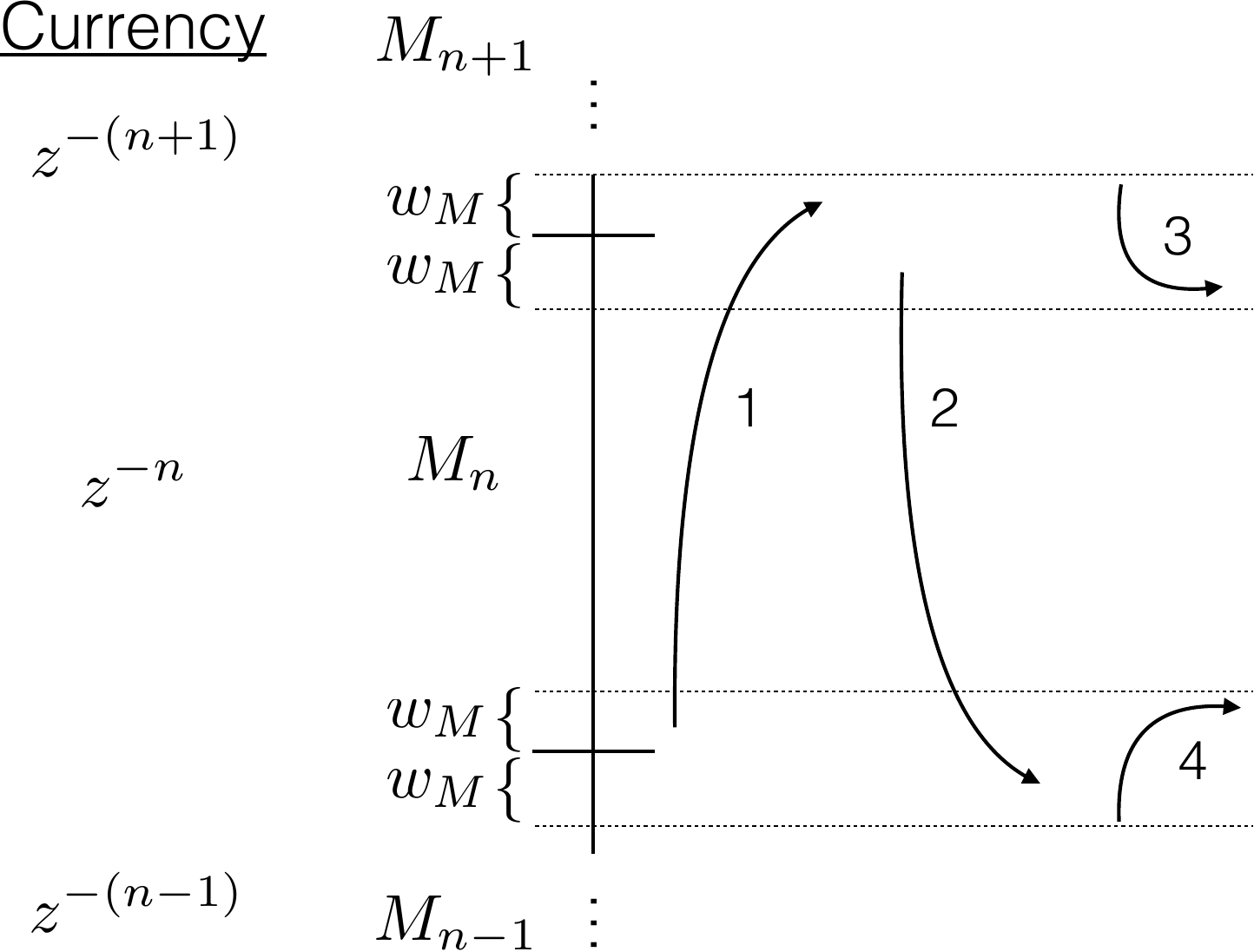}
\caption{An illustration of the different cases of changing currency. Dark lines mark the boundary of an energy block and dotted lines mark a region of size $E_M$ around the boundary.}
\label{fig:cases-recurrent}
\end{figure}
\end{proof}

It is not hard to show that Lemma~\ref{lem:max-invariant} implies that $f_M$ is legal. That is, consider a finite play $\pi$ that starts immediately after a change in currency. Using Lemma~\ref{lem:magic-max}, we can prove by induction on the length of $\pi$ that \Max has sufficient budget for bidding. The harder case is when $\pi$ decreases, and the proof follows from the fact that $wiggle$ is in the higher currency of the lower block. Combining Lemma~\ref{lem:max-invariant} with our choice of the initial energy, we get that the energy never reaches $0$ as otherwise \Min invests a budget of more than $1$. The following theorem follows by combining with Lemma~\ref{lem:simplify-max}.

\begin{theorem}
In a recurrent mean-payoff bidding game $\G$ with $\MP(\RT(\G))> 0$, with any positive initial budget \Max has a strategy that guarantees a positive payoff.
\end{theorem}

\subsubsection{An optimal \Max strategy in general strongly-connected mean-payoff bidding games}
In this section we develop the ideas of the previous section and construct an optimal strategy for \Max in general strongly-connected games. Recall that by Lemma~\ref{lem:simplify-max} it suffices to construct a strategy that guarantees that the energy is bounded from below. The following example shows that naively adapting the strategy from the previous section fails.

\begin{figure}[ht]
\centering
\includegraphics[height=7cm]{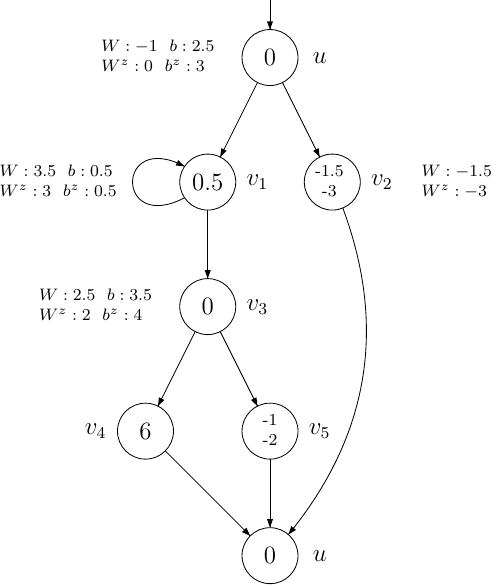}
\caption{An example showing that the \Max strategy developed in the previous section fails in general strongly-connected games. In a vertex $v$ with negative weight, the weight $w(v)$ of $v$ is depicted on top and $w^z(v)$ on the bottom. We choose $z=2$.}
\label{fig:counterexample}
\end{figure}
\begin{example}
\label{ex:counterexample}
Consider the strongly-connected mean-payoff bidding game $\G$ that is depicted in Figure~\ref{fig:counterexample}. Note that $\G$ is not recurrent. Indeed, the candidates for the root would be $u$ and $v_1$ and there are cycles that avoid both of them. We choose $z = 2$. With this choice, we have $v_1^+ = v_1$. Indeed, $\Po^z(v_1) > \Po^z(v_3)$. Thus, according to the strategy in the previous section, upon winning a bidding in $v_1$, \Max chooses the self-loop to stay in $v_1$. Since the weight of $v_1$ is positive, staying in $v_1$ implies an increase of energy, which implies a decrease of budget. Since \Max avoids exhausting his budget, the currency must change in $v_1$. In other words, \Max cannot wait for a visit to $u$ to change the currency. The inability to wait for visits to the root is the challenge of devising a strategy in general strongly-connected games.

A naive solution would be to drop the assumption from the previous section that currency changes occur only at $u$. That is, we change currency upon entering a new energy block no matter what the current vertex is. We illustrate that this attempt fails, implying that a more involved adaptation is needed. The problem is with sinusoidal energy behaviors that occur on the boundary of an energy block. We describe such a play. Consider the cycle $u, v_1, v_1, v_3, v_5, u$, which intuitively corresponds to \Max winning two biddings, then losing two biddings, and we ignore $v_5$ since both players bid $0$. In $\G^z$, we have equality between energy and budget. Indeed, we have $\St^z(u) + \St^z(v_1) - \St^z(v_1) - \St^z(v_3) = 3 + 0.5 - 0.5 - 4 = -1 = w^z(u) + w^z(v_1) + w^z(v_1) + w^z(v_3) + w^z(v_5)$. Note that the ``real'' energy is the one in $\G$ and it is unchanged following this path since $2w(v_1) + w(v_5) = 0$.

Assume we start from $u$ when the current energy is at the top of the third energy block. Recall that $z= 2$, thus the initial currency is $z^{-3} = 1/8$. After visiting $v_1$ twice, the energy increases and enters the fourth block, thus the currency is updated to $1/16$. Adding the currencies to the calculation above, we get $\St^z(u)\cdot z^{-3} + \St^z(v_1)\cdot z^{-3} - \St^z(v_1)\cdot z^{-4} - \St^z(v_3)\cdot z^{-4} = 3\cdot 1/8 + 0.5 \cdot 1/8- 0.5 \cdot 1/16 - 4 \cdot 1/16 > 0$. All in all, \Max's payments are positive, thus his budget decreases, while the energy level stays the same. \Min can thus continue with such a strategy until \Max's budget is exhausted.\hfill\qed
\end{example}

We develop further the ingredients from the previous sections. Recall that in recurrent games, we split the natural numbers into energy blocks, each block has a currency, where increasing the energy by $c$ units in the $n$-th block costs \Max at most $c$ units of budget in the currency of the $n$-th block, and decreasing the energy by $c$ units in the $n$-th block rewards him with at least $c$ units of budget in the higher currency of the $(n-1)$-th block. In general strongly-connected games, we need stronger properties. First, we increase the asymmetry between investing and gaining: while gaining in the $n$-th block is still in the higher currency of the lower $(n-1)$-th block, investing is now in the lower currency of the higher $(n+1)$-th block. Thus, now, in every change to the energy within an energy block, \Max registers a profit. The larger the change in energy, the larger the profit. Second, we differentiate between even blocks and odd blocks so that odd blocks serve as ``buffers'' that ensure that a change in currency only occurs after a significant change in energy. 

We formalize this intuition. Consider a strongly-connected mean-payoff bidding game $\G = \zug{V, E, w}$ having $\MP(\RT(\G)) > 0$. For $z > 1$, let $\tilde{\G}^z = \zug{V, E, \tilde{w}^z}$, where
\[
\tilde{w}^z(v) = \begin{cases} w(v) \cdot z & \text{ if } w(v) < 0 \\
w(v) \cdot \frac{1}{z} & \text{ if } w(v) \geq 0\end{cases}\]
As in the previous section, it is not hard to choose $z > 1$ such that $\MP(\RT(\tilde{\G}^z)) > 0$. Let $\tilde{\Po}^z$ and $\tilde{\St}^z$ denote the potentials and strengths in $\tilde{\G}^z$, and for a finite play we denote by $\tilde{E}^z$ be the sum of weights that $\pi$ traverses in $\tilde{\G}^z$. The proof of the following lemma is similar to Lemma~\ref{lem:E-Ez}.

\begin{lemma}
\label{lem:general-Max-connecting}
Consider a finite play $\pi$. We have $\tilde{E}^z(\pi) \leq z\cdot E(\pi)$ and $\tilde{E}^z(\pi) \leq\frac{1}{z} \cdot E(\pi)$. 
\end{lemma}

We describe \Max's strategy, which we refer to as $f_M$. As in the previous section, \Max chooses a $k_I \in \Nat$ and plays as if that is the initial energy while guaranteeing that the energy never reaches $0$. We specify $k_I$ later. When reaching a vertex $v \in V$, \Max bids $\tilde{\St}^z(v) \cdot \gamma$ and moves to $v^+$ upon winning, where we define the currency $\gamma \in (0,1)$ next. We partition $\Nat$ into blocks of $\tilde{N} \in \Nat$, where we choose $\tilde{N}$ later on. We refer to the $n$-th block as $\tilde{N}_n = \set{\tilde{N}\cdot (n-1),\ldots, \tilde{N}n-1}$. Unlike the previous section, changes in currency can occur in all vertices and only depend on the energy. The currency in even and odd blocks differs. For $n \in \Nat$, when the energy level reaches an even block $\tilde{N}_{2n}$,  the currency is $z^{-n}$. In order to determine the currency in the odd blocks, we take the history of the play into account; the currency matches the currency in the last energy block that was visited before entering $\tilde{N}_{2n+1}$. Thus, if it is $\tilde{N}_{2n}$, then the currency is $z^{-n}$ and if it is $\tilde{N}_{2n+2}$, the currency is $z^{-(n+1)}$. We say that a finite outcome is {\em $\gamma$-consistent} when all the bids \Max performs in it are made in the same currency $\gamma$. Lemma~\ref{lem:magic-max} clearly applies to $\tilde{\G}^z$. Let the maximal weight of a vertex in $\G$ be $w_M = \max_{v \in V} |W(v)|$. The following lemma follows from combining Lemma~\ref{lem:magic-max} with Lemma~\ref{lem:general-Max-connecting}.

\begin{lemma}
\label{lem:magic-Max-general}
Consider a $(z^{-n})$-consistent outcome $\pi$ that starts in $v$ and ends in $v'$. We have $-pay(\pi) \geq -E(\pi) \cdot z^{-(n-1)} - 2w_M\cdot z^{-n}$ and $pay(\pi) \leq  E(\pi) \cdot z^{-(n+1)} + 2w_M \cdot z^{-n}$. 
\end{lemma}

%Consider an initial \Max budget of $B^I_M$. We choose an even $k_I$ such that $B^I_M > (2w_M + b_M) \cdot z^{-k_I} + \sum_{i=k_I}^\infty (z^{-i} \cdot M + z^{-(i+1)}\cdot 2M)$. We show that assuming an initial energy of $k_I$, then by following $f_M$, \Max guarantees a non-negative energy level. 

Suppose \Max is playing according to $f_M$ and \Min is playing according to some strategy $f_m$. Let $\pi = play(f_m, f_M)$ be the resulting infinite play. Let $\pi = \pi^1 \cdot \pi^2 \cdot \ldots$ be a partition of $\pi$ into maximal finite plays that have a consistent currency. For $i \geq 1$, let $e^i \in \Nat$ be the energy at the end of $\pi^i$, thus $e^i = k_I + E(\pi^1\ldots\pi^i)$, where $k_I$ is an initial energy. Also, let $\beta^\uparrow_n$ and $\beta^\downarrow_n$ be respectively, the upper and lower boundaries of the energy block $N_n$. Note that $\beta^\uparrow_n = \beta^\downarrow_{n+1}$. 

\begin{figure}[ht]
\centering
\includegraphics[height=5cm]{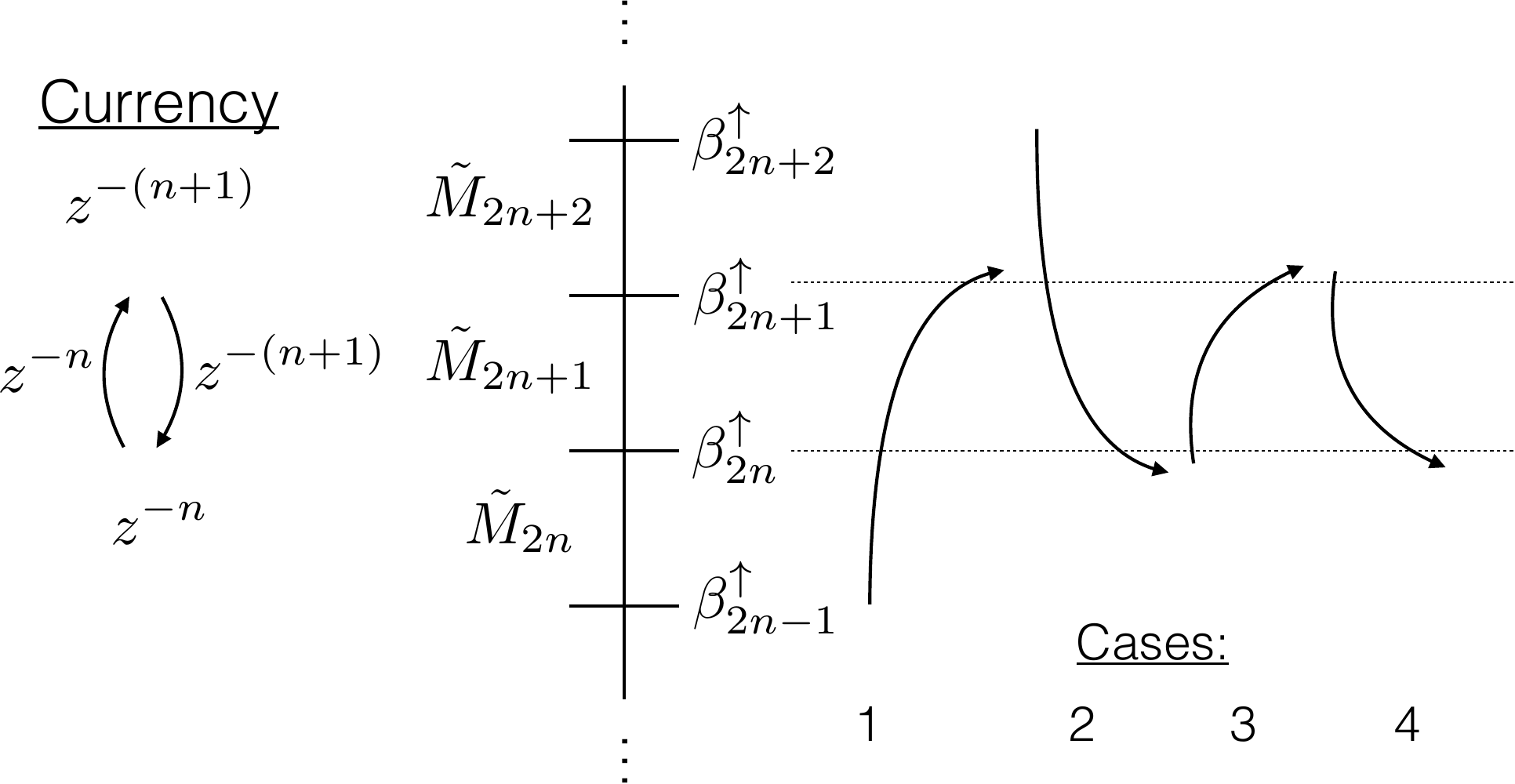}
\caption{The four cases of $\pi^i$ in the general setting.}
\label{fig:cases-general}
\end{figure}

Suppose a sub-play $\pi^i$ starts in a vertex $v$ and ends in $u$. We make observations on the budget change during $\pi^i$. There are four cases, which are depicted in Figure~\ref{fig:cases-general}. Note that the currency in Cases~$1$ and~$3$ is $z^{-n}$ and in Cases~$2$ and~$4$ it is $z^{-(n+1)}$. The energy change in $\pi^i$ in Cases~$1$ and~$2$ is at least $2\tilde{N}$ and at most $2\tilde{N}+2w_M$ and in Cases~$3$ and~$4$ it is at least $\tilde{N}$ and at most $\tilde{N}+2w_M$. We use Lemma~\ref{lem:magic-Max-general} to obtain the following:

\begin{lemma}
\label{lem:cases}
The following bounds hold for the change in budget in the four cases depicted in Figure~\ref{fig:cases-general}.
\begin{enumerate}
\item $pay(\pi^i) \leq (2\tilde{N}+2w_M)\cdot z^{-(n+1)} + 2w_M \cdot z^{-n}$,
\item $- pay(\pi^i) \geq 2\tilde{N}\cdot z^{-n} - 2w_M \cdot z^{-(n+1)}$,
\item $pay(\pi^i) \leq (\tilde{N} + 2w_M)\cdot z^{-(n+1)} + 2w_M \cdot z^{-n}$, and
\item $-pay(\pi^i) \geq \tilde{N}\cdot z^{-n} + 2w_M \cdot z^{-(n+1)}$.
\end{enumerate}
\end{lemma}

To conclude the construction, given an initial \Max budget, we find an initial energy level $k_I$ with which \Max can guarantee that the energy stays positive. We do this by finding an invariant on his budget at the end points of energy blocks. Recall the intuition that \Max's budget should not run out even when the energy increases arbitrarily. We thus require his initial budget to be sufficient to ``purchase'' all the energy blocks above the initial energy. For $n \in \N$, the cost of the blocks $\tilde{N}_{2n}$ and $\tilde{N}_{2n+1}$ is $\tilde{N} \cdot z^{-n}$. 

Recall from the previous section that \Max's budget at the bottom of an energy block $\tilde{N}_{2n}$, needs to include, in addition to the costs of the energy blocks $\tilde{N}_{\geq 2n}$, {\em wiggle room} in the currency of the lower block. Going back to Lemma~\ref{lem:cases}, we observe that Case~$4$ is the only problematic case. Indeed, in all other cases, the path $\pi^i$ crosses an energy block whose cost is given in a currency that is lower than the currency of gaining (when decreasing), or higher than the currency of investing (when increasing). Take for example Case~$2$. It crosses both $\tilde{N}_{2n+2}$ and $\tilde{N}_{2n+1}$. The cost of $\tilde{N}_{2n+2}$ is $\tilde{N} \cdot z^{-(n+1)}$ whereas the gain for it is roughly $\tilde{N} \cdot z^{-n}$. The situation in Case~$4$ is not that bad. The gain equals the cost of $\tilde{N}_{2n+1}$, i.e., $\tilde{N}\cdot z^{-n}$, up to a constant, i.e., $2w_M \cdot z^{-n}$. We add this constant in the invariant, thus we require \Max's budget at $\beta^\uparrow_{2n+1}$ to include the costs of the higher blocks, the wiggle room, and a surplus of $2w_M \cdot z^{-n}$. 

We define the invariant on \Max's budget formally. Recall that $wiggle = 2w_M + \tilde{\St}^z_M$, where $\tilde{\St}^z_M$ is the maximal bid, and it is used to guarantee that $f_M$ is legal in a play that stays in an energy block. We write $Inv(\beta^\uparrow_\ell)$ to refer to the budget that \Max has when the currency changes near $\beta^\uparrow_\ell$, thus within $|w_M|$ of $\beta^\uparrow_\ell$. We have the following.
\begin{itemize}
\item $Inv(\beta^\uparrow_{2n}) = wiggle \cdot z^{-n} + z^{-n}\tilde{N} +\sum_{i=n+1}^\infty 2z^{-i}\tilde{N}$, and
\item $Inv(\beta^\uparrow_{2n+1}) = wiggle \cdot z^{-n} + 2w_M\cdot z^{-(n+1)} + \sum_{i=n+1}^\infty 2z^{-i}\tilde{N}$.
\end{itemize}

To conclude the construction, we choose $\tilde{N}$ to be large enough so that the invariant is maintained assuming it is maintained initially. Also, given an initial budget for \Max, we choose an initial energy level such that the invariant is initially maintained. Combining with Lemma~\ref{lem:simplify-max}, we obtain the second direction of Theorem~\ref{thm:MP-SCC}.

\begin{theorem}
\label{thm:MP-max}
Consider a strongly-connected mean-payoff bidding game $\G$ with $\MP(\RT(\G)) > 0$. Then, \Max has a strategy that guarantees a positive payoff in $\G$.
\end{theorem}

\subsection{Remarks}
\subsubsection{Results for other bidding mechanisms}
\label{sec:poorman}
We elaborate on further results on infinite-duration bidding games that were obtained since an earlier publication of this paper. The bidding mechanism that we study in this paper is called {\em Richman bidding}. {\em Poorman bidding} is the same as Richman bidding only that the winner of the bidding pays the ``bank'' rather than the other player. {\em Taxman bidding} span the spectrum between Richman and poorman bidding. It is parameterized by a constant $\tau \in [0,1]$: portion $\tau$ of the winning bid is paid to the other player, and portion $1-\tau$ to the bank. Richman bidding is obtained by setting $\tau=1$ and poorman bidding by setting $\tau=0$. Unlike Richman bidding, in both of these mechanisms, the sum of budgets is not constant throughout the game. The central quantity that is studied is thus the {\em ratio} of the players' budget: suppose that for $i \in \set{1,2}$, \PLi's budget is $B_i$, then \PO's ratio is $B_1/(B_1 + B_2)$. Note that \PO's ratio coincides with his budget in Richman bidding. For qualitative games, the central question is the existence of a {\em threshold ratio}, which is the straightforward adaptation of the threshold budgets we use (see Definition~\ref{def:thresh}). 

Reachability games with poorman and taxman bidding have been studied in \cite{LLPSU99}. It is shown that while threshold ratios exist in reachability poorman and taxman games, the structure of the game is more complicated and no probabilistic connection is known and it is unlikely to exist: already in the reachability game that is depicted in Figure~\ref{fig:reach}, the threshold ratios with poorman bidding are irrational numbers. 

Infinite-duration bidding games with poorman bidding were studied in \cite{AHI18} and with taxman bidding in \cite{AHZ19Arxiv}. Given the probabilistic connection for reachability Richman-bidding games (Theorem~\ref{thm:reach}), the probabilistic connection for mean-payoff Richman-bidding games (Theorem~\ref{thm:MP-SCC}) may not be unexpected. On the other hand, since no probabilistic connection is known for reachability poorman-bidding games, we find the following probabilistic connection for mean-payoff poorman- and taxman-bidding games surprising. The ideas that were developed in the constructions in this paper played a key role in the proof of the following theorem.
\begin{theorem}\label{thm:poorman}
\cite{AHI18,AHZ19Arxiv}
Consider a strongly-connected mean-payoff taxman game $\G$ and a constant $\tau \in [0,1]$. The optimal payoff \Max can guarantee with an initial ratio $r \in (0,1)$ in $\G$ equals the value of the biased random-turn game $\RT^{F(\tau,r)}(\G)$, for $F(\tau,r) = \frac{r+\tau\cdot (1-r)}{1+\tau}$, in which in each turn \Max is chosen with probability $F(\tau, r)$ and \Min with probability $1-F(\tau,r)$. In particular, for poorman bidding, the optimal payoff in $\G$ with initial ratio $r$ equals $\MP(\RT^r(\G))$. 
\end{theorem}
Theorem~\ref{thm:poorman} sheds new light on Theorem~\ref{thm:MP-SCC}. Richman bidding is the exception of taxman bidding: For every $\tau < 1$, the optimal payoff depends both on the structure of the game and the initial ratio. Only in Richman bidding does the optimal payoff depend only on the structure of the game and not on the initial ratio. For example, recall that in the game that is depicted in Figure~\ref{fig:loops}, with Richman bidding, \Min can guarantee a non-positive payoff no matter what positive initial budget he starts with (using the tit-for-tat strategy for example). With poorman bidding, on the other hand, when \Max's initial budget is $2$ and \Min's initial budget is $1$, \Max's initial ratio is $\frac{2}{3}$, and the optimal payoff that \Max can guarantee is $\frac{2}{3}\cdot 1 + \frac{1}{3} \cdot (-1) = \frac{1}{3}$. Theorem~\ref{thm:poorman} implies an interesting connection between Richman and poorman bidding: the value in a mean-payoff bidding game with Richman bidding equals the value with poorman bidding and ratio $0.5$.

\subsubsection{An existential proof of Theorem~\ref{thm:MP-SCC}}
\label{sec:existential}
We describe an alternative existential proof of Theorem~\ref{thm:MP-SCC} that relies on a combination of the probabilistic connection for reachability bidding games that are played on infinite graphs \cite{LLPSU99} and results on probabilistic models \cite{BBEK11,BB+10}. The draw-back of this proof is that it does not give any insight on how to construct optimal strategies. That is, given a strongly-connected mean-payoff bidding game $\G$, using Theorem~\ref{thm:MP-SCC} and the existential proof, the only knowledge we obtain is the optimal payoff a player can guarantee. There is no hint, however, on how to construct a strategy that achieves this payoff, which, as can be seen in the previous sections, can be a challenging task. 

\begin{proof}[Existential proof of Theorem~\ref{thm:MP-SCC}]
Consider a strongly-connected mean-payoff bidding game $\G = \zug{V, E, w}$, where $w: V \rightarrow \Nat$. The {\em one-counter game}\footnote{Sometimes called an {\em energy game} \cite{BFLMS08}.} that corresponds to $\G$, denoted $\OCG(\G)$, is played on the same graph only with a different objective: a counter tracks the energy in an infinite play $\pi$, and $\pi$ is winning for \Min iff there exists a finite prefix in which the energy is $0$. That is, \Max wins $\pi$ iff the energy stays positive in every finite prefix of $\pi$. A {\em configuration} of $\OCG(\G)$ is a pair $\zug{v, n} \in V \times \Nat$, which intuitively means that the token is placed on $v$ and the accumulated energy (the counter value) is $n$. Lemmas~\ref{lem:simplify-Min} and~\ref{lem:simplify-max} can be rephrased to show the following correspondence between winning in the one-counter bidding game $\OCG(\G)$ and guaranteeing an optimal payoff in $\G$:

\noindent{\bf Claim:} If the threshold budget in every configuration $\zug{v,n}$ in $\OCG(\G)$ is $0$, i.e.,  \Min wins with any positive initial budget, then with every positive initial budget, \Min guarantees a non-positive payoff in $\G$. On the other hand, if for every vertex $v \in V$ and a positive initial budget $B_M > 0$ of \Max there is an initial energy $n \in \Nat$ such that $B_M > 1-\thresh(\zug{v,n})$ in $\OCG(\G)$, i.e., \Max can prevent \Min from winning when the game starts from $\zug{v,n}$, then \Max can guarantee a positive payoff in $\G$. 

The game $\OCG(\G)$ is a reachability bidding game that is played on an infinite graph. Formally, we have $\OCG(\G) = \zug{V \times \Nat, E', T}$, where $\zug{v', n'}$ is a neighbor of a vertex $\zug{v,n}$ iff $\zug{v,v'} \in E$ and the update to the counter is correct and stays non-negative, i.e., $n' = n + w(v)$ if $n' \geq 0$ and $n'=0$ otherwise, and the target for \Min is the set of vertices $V \times \set{0}$. A key property of this game is that even though the graph is infinite, the number of outgoing edges from each vertex is at most $|E|$ and in particular finite. The proof in \cite{LLPU96} of the probabilistic connection for reachability bidding games (Theorem~\ref{thm:reach}) extends to reachability games on infinite graphs in which all vertices have a finite out-degree. Thus, we have the following.

\noindent{\bf Claim:} The games $\OCG(\G)$ and $\RT(\OCG(\G))$ are equivalent: the threshold budget in a configuration $\zug{v,m} \in V \times \Nat$ in $\OCG(\G)$ equals the value of $\zug{v,n}$ in $\RT(\OCG(\G))$, i.e., the probability of winning under optimal play.

The game $\RT(\OCG(\G))$ is a stochastic game with a one counter. Such games have been shown to have the following properties.

\noindent{\bf Claim:} \cite{BBEK11,BB+10} When $\MP(\RT(\G)) \leq 0$, the value of every configuration $\zug{v,n}$ in $\RT(\OCG(\G))$ is $0$. When $\MP(\RT(\G)) > 0$, for every $v \in V$, the sequence $val(\RT(\OCG(\G)), \zug{v,n})$ tends to $0$ as $n$ tends to infinity. 

The proof of the theorem follows from combining the three claims.
\end{proof}

\subsubsection{Strategy complexity}
In this section we discuss the memory requirements of the strategies that we construct for mean-payoff bidding games, which we call the {\em complexity} of the strategy. The complexity of a strategy is important since strategies are typically used to implement systems, and the complexity of the strategy translates to the complexity of the system. In all three strategies, when the token is placed on a vertex $v$, the strategy always prescribes the same vertex to move to upon winning the bidding, namely $v^-$ for \Min and $v^+$ for \Max, and the bid is of the form $\St(v) \cdot \gamma$, where $\St(v)$ is a constant and $\gamma$ is the normalization factor, which changes as the game proceeds. Thus, a strategy uses memory only for determining the normalization factor. 

In \Min's strategy, recall that $\gamma$ is of the form $1/N$, where $N \in \Nat$ is chosen immediately after the energy hits $0$. To compute the normalization, \Min's strategy uses two variables that take integer values. One keeps track of the current energy level in order to observe that it hits $0$ and that a new $N$ needs to be chosen. The second variable keeps the current choice of $N$. 

In \Max's strategy in recurrent games, the normalization, which is called the currency of the energy block, changes in the root vertex of the game depending on the energy level. \Max's strategy again uses two variables that take integer values. The first keeps track of the energy and the second keeps the index of the energy block in the last visit to the root. In a vertex that is not the root, \Max computes the normalization by referring to the stored index of the energy block.

Finally, in \Max's strategy in general strongly-connected games, the normalization changes when the energy visits an even energy block. Again, \Max's strategy can be implemented using two variables that keep track of the current energy and the index of the energy block whose currency is currently being used.

\section{Discussion and Future Directions}
\label{sec:disc}
We introduce and study infinite-duration bidding games in which the players bid for the right to move the token. 
We showed the existence of threshold budgets in parity bidding games by reducing them to reachability bidding games. We also showed the existence of threshold budgets in mean-payoff bidding games. The key to the qualitative solution was a quantitative solution to strongly-connected mean-payoff bidding games: we showed that these games are equivalent to uniform random-turn games in the sense that the optimal payoff a player can guarantee in the bidding game equals the expected payoff in the stochastic game with optimal play. Thus, we show that the initial budgets do not matter in mean-payoff bidding games with the bidding rules we use, namely Richman bidding. That is, the payoff depends only on the structure of the game and not on the initial budgets. As we elaborate in Section~\ref{sec:poorman}, this is not the case with other bidding mechanisms, where the payoff depends both on the structure of the game and the initial budgets. 

This work belongs to a line of works that transfer concepts and ideas between the areas of formal verification and algorithmic game theory \cite{NRTV07}. Examples of works in the intersection of the two fields include logics for specifying multi-agent systems \cite{AHK02,CHP10,MMPV14}, studies of equilibria in games related to synthesis and repair problems \cite{CHJ06,Cha06,FKL10,AAK15}, non-zero-sum games in formal verification \cite{BBMU15,BBPG12,CMJ04}, and applying concepts from formal methods to {\em resource allocation games} such as rich specifications \cite{AKT16}, efficient reasoning about very large games \cite{AGK18,KT17}, reasoning about resource interfaces \cite{CA+03}, and a dynamic selection of resources \cite{AHK16}. 

We discuss some directions for future work. We studied the computational complexity of finding threshold budgets, which we formally define as the \THRESHBUD problem, and showed that for the objectives we consider, the problem is in NP and coNP using the reduction to random-turn games. We leave open the problem of finding a tighter classification for \THRESHBUD. Our result hints that the problem is not NP-hard. A tighter classification would be, optimistically, a polynomial-time algorithm for \THRESHBUD, or, pessimistically, showing that \THRESHBUD is as hard as solving general simple stochastic games, which is a problem in NP and coNP for which no polynomial-time algorithm is known.

In Section~\ref{sec:existential}, we discussed one-counter games in which \Min wins if the energy hits $0$ once in a play. Note that unlike parity and mean-payoff, this objective is not {\em prefix independent}. The complexity of \THRESHBUD in one-counter games is interesting and is related to recent work on optimizing the probability of reaching a destination in a weighted MDP \cite{HK15,RRS15}.  For acyclic one-counter bidding games, the problem is PP-hard using a result in \cite{HK15}, and for a single-vertex games the problem is in P using the direct formula of \cite{Kat13}. For general games the problem is open.

%of stam
\section*{Acknowledgments}
We thank Petr Novotn\'y and Rasmus Iben-Jensen for helpful discussions and pointers.% and Rasmus Iben-Jensen for showing us the counterexample for the local tit-for-tat strategy.

\small
\bibliographystyle{plain}
\bibliography{../ga}
\end{document}